\documentclass[aps,prb,floatfix,twocolumn,superscriptaddress]{revtex4-2}

\usepackage{mathrsfs}
\usepackage{graphicx}
\usepackage{amsmath}
\usepackage{amssymb}
\usepackage{amsthm}
\usepackage{tikz-cd}
\usepackage{xcolor}
\usepackage{bm}
\usepackage{hyperref}
\usepackage{capt-of}
\usepackage[english]{babel}
\usepackage{cancel}
\usepackage[normalem]{ulem}

\usepackage{epstopdf}
\usepackage{relsize}
\usepackage{CJK}
\usepackage{textcomp}
\usepackage{dsfont}
\usepackage{float}

\newcommand{\me}{\mathrm{e}}
\newcommand{\mi}{\mathrm{i}}
\newcommand{\dif}{\mathrm{d}}
\newcommand{\tr}{\operatorname{Tr}}
\newcommand{\rank}{\operatorname{rank}}

\newcommand{\End}{\operatorname{End}}

\newcommand{\diag}{\operatorname{diag}}

\newtheorem{proposition}{Proposition}

\begin{document}
\title{Uhlmann Geometry of Fixed-Rank Density Matrices: A Fiber-Bundle Approach}

\author{Xu-Yang Hou}
\affiliation{School of Physics, Southeast University, Nanjing 211189, China}

\author{Hao Guo}
\email{guohao.ph@seu.edu.cn}
\affiliation{School of Physics, Southeast University, Nanjing 211189, China}
\affiliation{Hefei National Laboratory, University of Science and Technology of China, Hefei 230088, China}

\begin{abstract}
For full-rank density matrices the state space is contractible, the Uhlmann
bundle is topologically trivial, and the holonomy admits no quantized
invariants; although the Uhlmann phase is genuinely geometric, this
topological poverty has kept it from serving as a robust physical diagnostic
of mixed-state matter. Mixed states of fixed rank below the Hilbert-space
dimension, however, are ubiquitous: reduced states of constrained subsystems, states confined to invariant
sectors, and states supported on decoherence-free subspaces can have a
support smaller than the Hilbert space, a support that can vary with
parameters. Their
geometry is far richer: the support carries genuine curvature, non-Abelian
holonomy, and Chern topology. We formulate Uhlmann's theory directly on the
manifold of rank-$k$ density matrices: minimal purifications make the
fixed-rank stratum the base of a principal $U(k)$-bundle whose Uhlmann
connection is uniquely determined by a Sylvester equation, and in an
eigenframe this connection takes a closed form reducing to the Berry,
Wilczek--Zee, and faithful Uhlmann connections at $k=1$, at equal weights,
and at $k=N$. The fixed-rank bundle inherits the topology of the
Grassmannian, and non-factorizable higher Chern topology requires failure of the global
eigenline splitting, which in the present spectral setting requires
degeneracy, minimally realized by a rank-2 Yang monopole whose quantized
second Chern number links the geometry to the four-dimensional quantum Hall
response. Solvable models, from a genuinely non-Abelian holonomy to a
dissipatively driven orbit whose Uhlmann holonomy is steered by the
reservoir through an elliptic-integral spectral dressing, illustrate the
content: the supporting subspace carries the topology, the spectral weights
shape the transport, and Uhlmann holonomy provides a direct geometric
probe of the resulting mixed-state structures.
\end{abstract}

\maketitle

\section{Introduction}
\label{intro}

Geometric phases have become a cornerstone of modern quantum physics,
providing a unifying language for topological phenomena that depend only on
the path traversed in parameter space rather than on the dynamical details of
the evolution. The concept originates from Berry's adiabatic phase for a
non-degenerate pure state \cite{Berry84}, which was soon generalized to
non-adiabatic evolution \cite{PhysRevLett.58.1593}. For a pure state the
natural geometric setting is the Hopf fibration
$U(1)\hookrightarrow S^{2N-1}\to\mathbb{C}P^{N-1}$, and the Berry phase
factor is the $U(1)$ holonomy of the associated connection;
this framework underlies a broad range of topological phenomena, from the
quantum Hall effect to topological insulators and superconductors
\cite{TKNN,Haldane,KaneRMP,ZhangSCRMP,MooreN,MoorePRB,FuLPRL,Bernevigbook,ChiuRMP,KaneMele,KaneMele2,BernevigPRL}.
When the relevant subspace is degenerate, the scalar Berry phase is replaced
by a non-Abelian holonomy, the Wilczek--Zee phase
\cite{PhysRevLett.52.2111}, which acts as a unitary matrix on the degenerate
subspace. This non-Abelian structure has found applications in molecular and
nuclear settings, including nuclear rotations of diatomic molecules
\cite{PhysRevLett.56.893,PhysRevLett.56.2779}, nuclear quadrupole resonance
\cite{PhysRevLett.60.2734,PhysRevA.42.3107}, dynamics of deformable bodies
\cite{Shapere_1989,shapere1989gauge}, molecular Kramers doublets
\cite{PhysRevLett.59.161}, semiconductor heterostructures
\cite{PhysRevB.57.12302}, and ion traps \cite{PhysRevA.59.2910}; it has also
been observed experimentally in atomic Bose-Einstein condensates
\cite{Sugawa_2021}, and a comprehensive account is given in
Refs.~\cite{Anandan_1988,Cheon_2009,PhysRevA.71.012110,PhysRevA.93.012116}.

Realistic quantum systems are, however, rarely isolated. They interact with an
environment and exist at finite temperature, so that the state is described by
a density matrix rather than a wavefunction. Generalizing geometric phases to
mixed states has therefore been a long-standing problem. Two main approaches
have been developed. The first, due to Sj\"oqvist and collaborators, defines
an interferometric geometric phase for mixed states through a phase shift
acquired by a Mach-Zehnder interferometer, and has been studied theoretically
\cite{PhysRevLett.85.2845,PhysRevA.67.020101,PhysRevLett.90.160402,Faria_2003}
and observed in nuclear magnetic resonance \cite{PhysRevLett.91.100403},
polarized neutrons \cite{PhysRevLett.101.150404}, and Mach-Zehnder
interferometry \cite{PhysRevLett.94.050401}. The second, formulated
independently by Uhlmann, is rooted in the purification of density matrices
\cite{Uhlmann1986,Uhlmann1989,Uhlmann1991,UHLMANN1995461}. For a density
matrix $\rho$, one introduces an amplitude $W$ such that $\rho=WW^\dagger$,
and a parallel-transport condition on $W$ defines a connection whose
holonomy gives the Uhlmann phase.

The Uhlmann geometric framework has become an important approach to mixed
quantum states. Early studies revealed finite-temperature phase transitions
through the Uhlmann phase and introduced Uhlmann numbers as geometric
indicators in one and two dimensions
\cite{Viyuela14,PhysRevLett.113.076408,PhysRevLett.113.076407}. It has since
been applied to various systems and phases, including composite,
superconducting, quantum-walk, nonequilibrium, and other condensed-matter
platforms
\cite{PhysRevA.98.033816,Zhang21,PhysRevA.104.042204,HeChien2022Lindblad,OurUhlmannQuench,PhysRevA.110.043313,PhysRevB.110.134319,PhysRevB.110.035144,ourPRB20,OurPRB20b}.
It has also been applied to complex multilevel systems, including three-level
and higher-spin models, revealing richer mixed-state geometric structures
\cite{GalindoRojasMaytorena2021,PhysRevA.104.023303,Hou2023,NIETOGUADARRAMA2024169706,98sq-16bz}.
Building on the Uhlmann phase, the connection has been used to construct
Chern-type quantities and study finite-temperature geometric and topological
properties of mixed states \cite{PhysRevB.97.235141,prq8-c9ns}.
Experimental realizations of Uhlmann-related quantities have been
demonstrated in quantum simulation platforms, including superconducting
qubits and photonic quantum walks
\cite{Viyuela2018TopologicalUhlmann,Mastandrea_2026,wang2025measuringmixedstatetopologicalinvariant}.
Parallel purification-based formulations have offered new perspectives on
Uhlmann anholonomy and its links to broader quantum holonomy structures
\cite{LevayVelich2026}.

Despite this progress, the standard formulation is most commonly presented
for faithful (full-rank) density matrices. Then $W$ is an invertible
$N\times N$ matrix, the purification bundle is a principal $U(N)$-bundle
over the faithful-state manifold, and the contractible base makes the bundle
topologically trivial; consequently no nontrivial characteristic class arises,
and the holonomy is unconstrained by integer invariants. For rank-deficient
states, the purification space is instead naturally described by minimal
$N\times k$ amplitudes, and the corresponding fixed-rank geometry requires a
separate treatment. The breakdown of the full-rank formalism near
rank-deficient states was analyzed in our previous work \cite{HuangDefect},
where the connection becomes singular on the boundary strata; here we provide
the constructive counterpart by formulating Uhlmann theory intrinsically on
the fixed-rank stratum. The fixed-rank purification bundle has also been noted
as geometric background in the information-geometry literature \cite{Mera19};
here we develop the corresponding theory systematically, including the
explicit connection, topology, and physical consequences.

This restriction is not merely technical. At finite temperature or in
open systems, density matrices can instead remain confined to a fixed-rank
sector below the Hilbert-space dimension: a reservoir may decohere a system
into a low-rank steady state, an exact zero-temperature projector may be the
relevant state, or an invariant or symmetry-protected subspace may be
populated while its complement is empty.
The state then lies on a fixed-rank stratum $\mathcal{D}_k^N$ with $k<N$,
whose topology differs qualitatively from the faithful-state manifold. The
base is no longer contractible: it deformation-retracts onto
$\operatorname{Gr}(k,N)$, the structure group reduces from $U(N)$ to $U(k)$,
and nontrivial characteristic classes become possible. Extending Uhlmann
theory to this fixed-rank setting is thus physically motivated and
mathematically natural.

In this work we construct such an extension. Minimal purifications
$W\in\mathbb{C}^{N\times k}$ place a principal $U(k)$-bundle over
$\mathcal{D}_k^N$, with the Uhlmann connection fixed uniquely by a Sylvester
equation. The eigenframe formula interpolates between the Berry connection
for $k=1$, the Wilczek--Zee connection for an equal-weight projector, and
the full-rank Uhlmann connection for $k=N$. We analyze the topology of the
fixed-rank bundle, show that its first Chern class equals that of the
tautological bundle, and prove that for non-degenerate spectra all
characteristic classes factorize into eigenline Berry data, so that
non-factorizable higher Chern topology requires failure of the global
eigenline splitting, which in the present spectral setting requires
degeneracy. We then
study four examples, including a rank-2 Yang monopole with quantized second
Chern number and a dissipatively driven qutrit orbit whose Uhlmann holonomy
records the reservoir's tracking ratio through a closed elliptic-integral
dressing factor, and clarify that the scalar Uhlmann phase is geometric
rather than topological, with integer topology residing in characteristic
classes and non-Abelian Wilson loops.

The rest of the paper is organized as follows. Section~\ref{fUB} recalls the
full-rank Uhlmann construction. Section~\ref{geometry} develops the
fixed-rank principal bundle framework, including minimal purifications, the
$U(k)$ connection, holonomy, and the Bures metric. Section~\ref{topology}
analyzes the topology of the fixed-rank bundle, its Chern classes, the
factorization for non-degenerate spectra, and the geometric versus
topological character of the Uhlmann phase. Section~\ref{examples} presents
four illustrative examples and compares them. Section~\ref{conclusion}
concludes. Detailed derivations are collected in the appendices.

\section{Overview of the Full-Rank Uhlmann Bundle}
\label{fUB}

Throughout this work we use natural units $\hbar=1$. We first recall the
full-rank Uhlmann construction, to highlight the modifications needed for
fixed rank below the Hilbert-space dimension.

For a full-rank density matrix $\rho>0$, a purification is an invertible
$N\times N$ matrix $W$ satisfying $\rho=WW^\dagger$. The polar decomposition
$W=\sqrt{\rho}\,U$ with $U\in U(N)$ exhibits the $U(N)$ gauge freedom of the
purification. The Uhlmann parallel-transport condition
\begin{equation}
W^\dagger \dot W = \dot W^\dagger W
\label{Uhlmann-cond-full}
\end{equation}
selects, among all purifications of a curve $\rho(t)$, the one that minimizes
the Hilbert--Schmidt distance between infinitesimally separated purifications
and is independent of the gauge choice $U$~\cite{Uhlmann86}. Substituting
$W=\sqrt{\rho}\,U$ into Eq.~\eqref{Uhlmann-cond-full} yields the
parallel-transport equation for the unitary factor,
\begin{equation}\label{NU}
\dot U + \mathcal A_{\rm U} U = 0,
\end{equation}
where the Uhlmann connection $\mathcal A_{\rm U}$ is given in the eigenbasis
of $\rho$ by
\begin{equation}
\mathcal{A}_{\rm U}
= -\sum_{a,b=1}^{N}
\frac{\langle a|\,[\mathrm{d}\sqrt{\rho}, \sqrt{\rho}]\,|b\rangle}
{\lambda_a+\lambda_b}\,|a\rangle\langle b|.
\label{eq:full-rank-AU}
\end{equation}
Here $\lambda_a>0$ are the eigenvalues of $\rho$ and $|a\rangle$ the
corresponding eigenvectors. It follows directly from
Eq.~\eqref{eq:full-rank-AU} that $\mathcal A_{\rm U}^\dagger=-\mathcal A_{\rm
U}$.

The base manifold of the full-rank theory, i.e., the manifold of faithful
states, is contractible; consequently, the associated principal $U(N)$-bundle
is topologically trivial. Our aim is to extend this construction to density
matrices of fixed rank $k<N$, where the base manifold becomes nontrivial and
the structure group reduces from $U(N)$ to $U(k)$.

\section{Geometric Setup: Principal Bundle Framework}\label{geometry}
\subsection{The fixed-rank manifold $\mathcal{D}_k^N$}

Let $\mathcal{H}=\mathbb{C}^N$ be the Hilbert space of the system. We
consider the set of density matrices of fixed rank $k$ ($1\le k\le N$):
\begin{equation}
\mathcal{D}_k^N=\left\{\rho\in\End(\mathcal{H})\;\middle|\;\rho\ge 0,\ \tr\rho=1,\ \rank\rho=k\right\}.
\end{equation}
This set is a smooth submanifold (a stratum) of the full density matrix
space. Its real dimension can be counted as follows: the support of $\rho$ is
a $k$-dimensional subspace $S\subset\mathcal{H}$, parameterized by the
complex Grassmannian $\operatorname{Gr}(k,N)$ of real dimension $2k(N-k)$. On
$S$, $\rho$ restricts to a strictly positive $k\times k$ matrix of unit
trace, which has real dimension $k^2-1$. Hence
\begin{equation}
\dim_{\mathbb{R}}\mathcal{D}_k^N=2k(N-k)+k^2-1=2Nk-k^2-1.
\label{dimD}
\end{equation}
For $k=N$ this gives $N^2-1$, the dimension of the faithful-state manifold.

It is useful to view the full density matrix space as a disjoint union of
strata:
\begin{equation}
\mathcal{D}^N=\coprod_{k=1}^N \mathcal{D}_k^N,
\qquad
\overline{\mathcal{D}_k^N}=\bigcup_{j=1}^k \mathcal{D}_j^N.
\end{equation}
The closure relation reflects the fact that a rank-$k$ density matrix can be
obtained as a limit of rank-$j>k$ states by letting $j-k$ eigenvalues tend to
zero. This stratified structure is important when one attempts to extend
geometric constructions across rank-changing paths; in the present work we
remain within a fixed stratum.
\subsection{Minimal purifications and the principal $U(k)$-bundle}

Instead of using singular $N\times N$ amplitudes, we introduce the space of
\emph{minimal purifications}
\begin{equation}
\mathcal{P}_k=\left\{W\in\mathbb{C}^{N\times k}\;\middle|\;\rank W=k,\ \tr W^\dagger W=1\right\}.
\label{eq:Pk}
\end{equation}
The projection map $\pi:\mathcal{P}_k\to\mathcal{D}_k^N$ is defined by
\begin{equation}
\pi(W)=WW^\dagger=\rho.
\end{equation}
The right action of $V\in U(k)$ on the total space, $W\mapsto WV$,
preserves the projection since $(WV)(WV)^\dagger=WW^\dagger$. Because $W$ has
full column rank, this action is free: $WV=W$ implies $V=I_k$. The
dimensions satisfy
\begin{align}
\dim_{\mathbb{R}}\mathcal{P}_k-\dim_{\mathbb{R}}U(k)&=(2Nk-1)-k^2=2Nk-k^2-1\notag
\\&=\dim_{\mathbb{R}}\mathcal{D}_k^N,
\end{align}
so that $\mathcal{P}_k$ has exactly the dimension of a principal $U(k)$-bundle
over $\mathcal{D}_k^N$.

To complete the bundle structure, we need to show that each fiber is a copy
of $U(k)$. This follows from the polar decomposition of rectangular matrices
developed in Sec.~\ref{IIC}: if
$W_1W_1^\dagger=W_2W_2^\dagger=\rho$, with both $W_1,W_2$ full-column-rank
$N\times k$ matrices, then there exists a unique $V\in U(k)$ such that
$W_2=W_1V$. Consequently,
\begin{equation}
U(k)\hookrightarrow\mathcal{P}_k\xrightarrow{\;\pi\;}\mathcal{D}_k^N
\label{principal_bundle}
\end{equation}
is a genuine principal bundle. Here $\mathcal{P}_k$ is the total space,
$\mathcal{D}_k^N$ is the base manifold, and each fiber $\pi^{-1}(\rho)$ over
$\rho\in\mathcal{D}_k^N$ is diffeomorphic to the structure group $U(k)$.

\subsection{Polar decomposition for rectangular amplitudes}
\label{IIC}

We now construct the rectangular analogue of the polar decomposition
$W=\sqrt{\rho}\,U$ used in the full-rank case. Let $W\in\mathcal{P}_k$ be a
minimal purification. Its singular value decomposition (SVD) reads
$W=X\Sigma \mathcal{V}^\dagger$, where $X\in\mathbb{C}^{N\times k}$ satisfies
$X^\dagger X=I_k$, $\Sigma\in\mathbb{R}^{k\times k}$ is positive diagonal, and
$\mathcal{V}\in U(k)$. Then $\rho=WW^\dagger=X\Sigma^2 X^\dagger$, and the
unique positive-semidefinite square root is
\begin{equation}
\sqrt{\rho}=X\Sigma X^\dagger,
\end{equation}
which vanishes on $\ker\rho$. The appropriate inverse is the
Moore--Penrose pseudoinverse
\begin{equation}
\sqrt{\rho^{+}}=X\Sigma^{-1}X^\dagger
\quad\bigl(\text{with }(\sqrt{\rho})^+=\sqrt{\rho^+}\bigr),
\end{equation}
satisfying
$\sqrt{\rho}\,\sqrt{\rho^{+}}=XX^\dagger\equiv P_\rho$, the orthogonal
projector
onto $\operatorname{supp}\rho$.

Recall that an operator $\mathcal{U}$ is a partial isometry if and only if
both
$\mathcal{U}^\dagger\mathcal{U}$ and $\mathcal{U}\mathcal{U}^\dagger$ are
orthogonal projectors. The left polar decomposition of $W$ is then
\begin{equation}
W = \sqrt{\rho}\,\mathcal{U},
\qquad
\mathcal{U} = \sqrt{\rho^{+}}W,
\label{eq:polar}
\end{equation}
and substituting the SVD gives $\mathcal{U} = X\mathcal{V}^\dagger$. Therefore
$\mathcal{U}$ is a partial isometry with $\mathcal{U}^\dagger\mathcal{U}=I_k$
and $\mathcal{U}\mathcal{U}^\dagger=P_\rho$. It is the unique partial
isometry
with range $\operatorname{supp}\rho$ appearing in this decomposition;
it is not an arbitrary $N \times N$ unitary matrix. For $k=N$, the partial
isometry becomes a full unitary and Eq.~\eqref{eq:polar} reduces to the
standard polar decomposition of the faithful theory.

The partial-isometry condition also exhibits the reduction of the structure
group from $U(N)$ to $U(k)$: the freedom in choosing $\mathcal{U}$ is
precisely a right multiplication by an element of $U(k)$, which corresponds
to the fiber of the bundle \eqref{principal_bundle}. Indeed, if
$W=\sqrt{\rho}\,\mathcal{U}_1=\sqrt{\rho}\,\mathcal{U}_2$ with both
$\mathcal{U}_1,\mathcal{U}_2$ partial isometries having range
$\operatorname{supp}\rho$, then $\sqrt{\rho}(\mathcal{U}_1-\mathcal{U}_2)=0$
forces $\mathcal{U}_1=\mathcal{U}_2$.

\subsection{Parallel transport and the Uhlmann connection}
\label{connection}

\subsubsection{Horizontal subspace and the Uhlmann condition}

We equip the total space $\mathcal{P}_k$ with the Hilbert--Schmidt metric
\begin{equation}
g_{\mathrm{HS}}(X,Y)=\operatorname{Re}\tr(X^\dagger Y),\quad X,Y\in T_W\mathcal{P}_k.
\end{equation}
The vertical subspace is $V_W\mathcal{P}_k=\{W\xi\mid \xi\in\mathfrak{u}(k)\}$,
where $\xi^\dagger=-\xi$. The horizontal subspace $H_W\mathcal{P}_k$ is its
$g_{\mathrm{HS}}$-orthogonal complement. Thus $X_H\in T_W\mathcal{P}_k$ is
horizontal iff for all $\xi\in\mathfrak{u}(k)$,
\begin{equation}
\operatorname{Re}\tr(X_H^\dagger W\xi)=0.
\end{equation}
Writing $M=X_H^\dagger W=M_H+M_A$ with $M_H^\dagger=M_H$ and $M_A^\dagger=-M_A$,
the Hermitian part drops out because $\tr(M_H\xi)$ is purely imaginary.
The remaining condition $\tr(M_A\xi)=0$ for all $\xi\in\mathfrak{u}(k)$ forces
$M_A=0$ by non-degeneracy of the trace form on $\mathfrak{u}(k)$. Hence
\begin{equation}
X_H^\dagger W = W^\dagger X_H,
\label{eq:horizontal}
\end{equation}
which is the Uhlmann horizontal condition, identical in form to the full-rank
parallel-transport condition Eq.~\eqref{Uhlmann-cond-full}. Along a curve
$W(t)$, with $X_H=\dot W$, this becomes
$W^\dagger\dot W=\dot W^\dagger W$.

\subsubsection{Ehresmann connection and the Sylvester equation}

Since $T_W\mathcal{P}_k=H_W\mathcal{P}_k\oplus V_W\mathcal{P}_k$, any tangent
vector $X\in T_W\mathcal{P}_k$ decomposes uniquely as
\begin{equation}
X=X_H+W\omega(X),\quad X_H\in H_W\mathcal{P}_k,\quad \omega(X)\in\mathfrak{u}(k).
\label{eq:decomp}
\end{equation}
The assignment $X\mapsto\omega(X)$ defines a $\mathfrak{u}(k)$-valued one-form
on $\mathcal{P}_k$. Imposing the horizontal condition on $X_H=X-W\omega(X)$,
\begin{equation}
W^\dagger\bigl(X-W\omega(X)\bigr)=\bigl(X-W\omega(X)\bigr)^\dagger W,
\end{equation}
using $\omega(X)^\dagger=-\omega(X)$,
leads to the Sylvester equation
\begin{equation}
h\,\omega(X)+\omega(X)\,h = W^\dagger X - X^\dagger W,
\quad h\equiv W^\dagger W>0.
\label{eq:sylvester}
\end{equation}
Because $h$ is strictly positive, the linear map
$L:\omega\mapsto h\omega+\omega h$ is invertible on $\mathfrak{u}(k)$;
this point is detailed in Appendix~\ref{app:connection}. Hence $\omega(X)$
exists uniquely for every tangent vector $X$.

We now state the three defining properties that make $\omega$ an Ehresmann
connection on the principal bundle \eqref{principal_bundle}.

(1) \emph{ Reproducing property.} For $u\in\mathfrak{u}(k)$, let $u^\#$ denote
the fundamental vector field generated by the right action:
\begin{equation}
u^\#_W=\frac{\dif}{\dif t}\Big|_{t=0} W \me^{tu}=Wu,
\label{eq:fundamental}
\end{equation}
which is vertical. Therefore Eq.~\eqref{eq:decomp} gives
\begin{equation}
\omega(u^\#)=u.
\end{equation}

(2) \emph{ Equivariance.} Under the right translation $R_V:W\mapsto WV$ with
constant $V\in U(k)$, the connection transforms as
\begin{equation}
R_V^*\omega = V^\dagger\omega V = \mathrm{Ad}_{V^{-1}}\omega.
\end{equation}

(3) \emph{Gauge transformation law.} For a local gauge transformation
$W'=WV$ with $V$ depending on the base point, one has
\begin{equation}
\omega' = V^\dagger\omega V + V^\dagger\dif V.
\label{eq:gauge}
\end{equation}

The proofs of the last two properties are given in
Appendix~\ref{app:connection}.
These three properties identify $\omega$ as an Ehresmann connection, and its
kernel is the horizontal distribution, $\ker\omega_W=H_W\mathcal{P}_k$.
Consequently, a curve $\tilde\gamma(t)=W(t)$ is a horizontal lift of its
projection $\gamma=\pi\circ\tilde\gamma$ if and only if $\omega(\dot W)=0$.
Equivalently, the Uhlmann condition is equivalent to the vanishing of the
connection:
\begin{equation}
W^\dagger\dot W=\dot W^\dagger W
\Longleftrightarrow
\omega(\dot W)=0.
\label{eq:horizontal_equivalence}
\end{equation}
This follows by contracting the decomposition \eqref{eq:decomp} with $\dot W$
and using the bijectivity of $L$; the details are also in
Appendix~\ref{app:connection}.

\subsection{Uhlmann connection and holonomy}
\label{direct_derivation}

\subsubsection{Local sections and the Uhlmann connection}

To obtain a computable expression for the Uhlmann connection, we choose a
local
orthonormal frame $\{|a\rangle\}_{a=1}^k$ spanning
$\operatorname{supp}\rho$,
assembled into $E=(|1\rangle,\dots,|k\rangle)\in\mathbb{C}^{N\times k}$ with
$E^\dagger E=I_k$. In this frame $\rho=ErE^\dagger$ with
$r=\diag(\lambda_1,\dots,\lambda_k)>0$. Such a diagonalizing frame exists on
any patch where the spectrum is non-degenerate; near degeneracies one may
always choose a smooth support frame, for which $r$ is merely positive
Hermitian, and Eq.~\eqref{AU_explicit} below is to be understood in a smooth
spectral frame wherever one exists, while Eq.~\eqref{AU_sylvester} remains
the coordinate-independent statement. A natural local section of the bundle
\eqref{principal_bundle} is
\begin{equation}
s(\rho)=E\sqrt{r}.
\end{equation}
Unlike the full-rank case $k=N$, where $\sigma(\rho)=\sqrt{\rho}$ is a global
section, $s$ is in general only locally defined because the eigenframe $E$ is
fixed only up to the allowed rotations and no smooth global choice exists when
the base $\mathcal{D}_k^N\simeq\operatorname{Gr}(k,N)$ is topologically
nontrivial. A detailed discussion is given in Appendix~\ref{app:local}.

Pulling back the Ehresmann connection $\omega$ by $s$ yields the local
$\mathfrak{u}(k)$-valued Uhlmann connection on the base manifold,
\begin{equation}
\mathcal{A}_{\text{U}}\equiv s^*\omega\in
\Omega^1\bigl(\mathcal{D}_k^N,\mathfrak{u}(k)\bigr).
\label{AU_def}
\end{equation}
Within the domain of $s$, any point of the fiber is written uniquely as
\begin{equation}
W=s(\rho)\mathcal{U},\qquad \mathcal{U}\in U(k),
\label{W=sV}
\end{equation}
generalizing the full-rank parametrization $W=\sqrt{\rho}U$. Here
$\mathcal{U}\in U(k)$ is the $k\times k$ fiber coordinate, which differs from
the $N\times k$ partial isometry in Eq.~\eqref{eq:polar}; the latter is
recovered as $E\mathcal{U}$ in the local eigenframe. Using the reproducing
property and equivariance of the Ehresmann connection
$\omega$ mentioned above, one obtains the local relation between $\omega$ and
the pulled-back
connection $\mathcal{A}_{\text{U}}$,
\begin{equation}
\omega=\mathcal{U}^\dagger \pi^*\mathcal{A}_{\text{U}}\mathcal{U}
+\mathcal{U}^\dagger \dif_P \mathcal{U},
\label{omegaV}
\end{equation}
generalizing the full-rank relation. The derivation is given in
Appendix~\ref{app:local}. Substituting
Eq.~\eqref{W=sV} and Eq.~\eqref{omegaV} into the Sylvester equation
\eqref{eq:sylvester}, the terms involving $\dif \mathcal{U}$ cancel by
unitarity of $\mathcal{U}$,
and one obtains
\begin{equation}
r\mathcal{A}_{\text{U}}+\mathcal{A}_{\text{U}}r
=s^\dagger\dif s-\dif s^\dagger s.
\label{AU_sylvester}
\end{equation}
The detailed derivation is also given in Appendix~\ref{app:local}. Evaluating
the right-hand side in the
eigenframe gives the explicit formula
\begin{equation}
(\mathcal{A}_{\text{U}})_{ab}
=\frac{2\sqrt{\lambda_a\lambda_b}}{\lambda_a+\lambda_b}
\langle a|\dif b\rangle.
\label{AU_explicit}
\end{equation}
This is the central result of the fixed-rank theory. It is manifestly
anti-Hermitian, $(\mathcal{A}_{\text{U}})_{ba}^*=-(\mathcal{A}_{\text{U}})_{ab}$.
Under a gauge transformation $W\to WV$ the connection transforms in the
standard way
\begin{equation}
\mathcal{A}_{\text{U}}\to V^\dagger \mathcal{A}_{\text{U}} V
+V^\dagger \dif V.
\label{gauge_AU}
\end{equation}
For $k=N$ the section $s=E\sqrt{r}$ is gauge-equivalent, through
$\sigma(\rho)=s(\rho)E(\rho)^\dagger=E\sqrt{r}E^\dagger=\sqrt{\rho}$, to the
global canonical section $\sigma=\sqrt{\rho}$, and Eq.~\eqref{AU_explicit}
reduces to the standard full-rank Uhlmann connection, confirming that the
present construction is a genuine extension.

\subsubsection{Parallel transport and Uhlmann holonomy}

Recall that horizontal transport along a curve $\tilde\gamma(t)=W(t)$ is
characterized by $\omega(\dot W)=0$, where $\dot W$ is the tangent to the
lift.
Let $\gamma(t)=\pi(W(t))$ be the projected curve with tangent
$X=\dot\gamma$, so that $\pi_*\dot W=X$. Using the relation (\ref{omegaV}),
we find
\begin{equation}
\omega(\dot W)
=\mathcal{U}^\dagger \mathcal{A}_{\text{U}}(X)\mathcal{U}
+\mathcal{U}^\dagger \dot{\mathcal{U}}
.
\end{equation}
Since $\mathcal{U}$ is unitary and hence invertible, the horizontal condition
$\omega(\dot W)=0$ is equivalent to
\begin{equation}
\nabla_X\mathcal{U}\equiv\frac{\dif \mathcal{U}}{\dif t}
+\mathcal{A}_{\text{U}}(X)\mathcal{U}=0.
\label{covariant_V}
\end{equation}
This is the fixed-rank counterpart of the covariant equation (\ref{NU})
obeyed by the
phase factor in the faithful theory. Integrating along a closed loop $C$
yields
\begin{equation}
\mathcal{U}(C)=\mathcal{P}\exp\left(-\oint_C
\mathcal{A}_{\text{U}}\right)\mathcal{U}(0),
\end{equation}
so that the Uhlmann holonomy is
\begin{equation}
\mathcal{U}_{\text{U}}(C)=\mathcal{P}\exp\left(-\oint_C
\mathcal{A}_{\text{U}}\right)
\in U(k).
\label{holonomy}
\end{equation}
The Uhlmann phase for a closed loop $C$ is then defined by
\begin{equation}
\theta_{\text{U}}=\arg\tr\left[W^\dagger(0)W(T)\right]
=\arg\tr_k\left[r(0)\,\mathcal{U}_{\text{U}}(C)\right],
\label{theta_U}
\end{equation}
where $r(0)=s^\dagger(\rho(0))s(\rho(0))$ is the $k\times k$ restriction of
$\rho(0)$ to its support in the chosen eigenframe. For $k=N$ this reduces to
the familiar expression $\arg\tr[\rho(0)\mathcal{U}_{\text{U}}(C)]$.

The corresponding curvature is
$\mathcal{F}_{\text{U}}=\dif \mathcal{A}_{\text{U}}
+\mathcal{A}_{\text{U}}\wedge \mathcal{A}_{\text{U}}$. For a contractible
loop, the non-Abelian Stokes theorem gives, schematically,
$\mathcal{U}_{\text{U}}(\partial S)=\mathcal{P}_S\exp\bigl(-\mathlarger{\int}_S \mathcal{F}_{\text{U}}\bigr)$.
Since $\mathcal{D}_k^N\simeq\operatorname{Gr}(k,N)$ is simply connected, a
flat
connection would imply trivial holonomy for all loops. All nontrivial phases
in
the examples below therefore arise from genuine curvature; the topological
information is carried by the integer-valued Chern classes of the bundle, as
discussed in Sec.~\ref{topology}.

\subsection{Bures metric and horizontal lifts}
\label{bures}

The Bures distance between two density matrices is defined through the
minimal
Hilbert--Schmidt distance between their purifications,
\begin{equation}
\dif s_B^2(\rho_1,\rho_2)
=\inf_{W_{1,2}}\tr\left[(W_1-W_2)(W_1-W_2)^\dagger\right].
\end{equation}
For $\rho_1,\rho_2\in\mathcal{D}_k^N$, the infimum may be restricted to
minimal purifications $W_1,W_2\in\mathbb{C}^{N\times k}$, since any
purification of a rank-$k$ state can be compressed to a minimal one without
increasing the distance. The infimum is attained when $W_1$ and $W_2$ are
\emph{parallel} in Uhlmann's sense, i.e.,
$W_1^\dagger W_2=W_2^\dagger W_1\ge0$, yielding the Uhlmann fidelity
\begin{equation}
F(\rho_1,\rho_2)=\tr\sqrt{\sqrt{\rho_1}\rho_2\sqrt{\rho_1}},
\end{equation}
and $\dif s_B^2=2-2F$. The detailed derivation is given in
Appendix~\ref{app:bures}.

For an infinitesimal displacement $\dif\rho$, the Bures line element on the
fixed-rank stratum is
\begin{equation}
\dif s_B^2
=\frac{1}{2}\sum_{\lambda_i+\lambda_j>0}
\frac{|\langle i|\dif\rho|j\rangle|^2}{\lambda_i+\lambda_j}.
\label{bures_metric}
\end{equation}
Here the sum includes all eigenstates with $\lambda_i+\lambda_j>0$, i.e.,
both
support-support and support-kernel blocks. The kernel-kernel block is
excluded
because rank-changing directions are orthogonal to $\mathcal{D}_k^N$. The
support-kernel contribution is essential: support variations produce nonzero
matrix elements between support and kernel eigenvectors, encoding the full
Grassmannian geometry. The derivation of Eq.~\eqref{bures_metric} and its
equivalence with the minimal purification distance are given in
Appendix~\ref{app:bures}.

Equivalently, for a curve $\gamma(t)\subset\mathcal{D}_k^N$ with horizontal
lift $\tilde\gamma(t)\subset\mathcal{P}_k$,
\begin{equation}
\int_\gamma \dif s_B
=\int_{\tilde\gamma}\sqrt{\operatorname{Re}\tr(\dot W^\dagger \dot W)}\,\dif t,
\end{equation}
confirming compatibility of the fixed-rank Uhlmann bundle with the metric
structure of mixed states.

\subsection{Reductions and limiting cases}
\label{reductions}

\subsubsection{Rank-one: Berry phase}

For $k=1$, the density matrix is pure, $\rho=|\psi\rangle\langle\psi|$. The
only eigenvalue is $\lambda_1=1$, and Eq.~\eqref{AU_explicit} gives
\begin{equation}
\mathcal{A}_{\text{U}}=\langle\psi|\dif\psi\rangle=\mathcal{A}_{\mathrm{Berry}}.
\end{equation}
The principal bundle \eqref{principal_bundle} becomes the Hopf bundle
$U(1)\hookrightarrow S^{2N-1}\to\mathbb{C}P^{N-1}$, and the Uhlmann phase
reduces exactly to the Berry phase.

To see this explicitly, note that for $k=1$ the minimal purification
$W\in\mathbb{C}^{N\times 1}$ is a column vector $|\psi\rangle$, the
eigenframe
$E$ reduces to $|1\rangle=|\psi\rangle$, and $r=1$, $\sqrt{r}=1$. The frame
connection is $K=\langle\psi|\dif\psi\rangle$. Inserting these into
Eq.~\eqref{AU_explicit}, the only component is
\begin{equation}
(\mathcal{A}_{\text{U}})_{11}
= \frac{2\sqrt{\lambda_1\lambda_1}}{\lambda_1+\lambda_1}
\langle 1|\dif 1\rangle
= \langle\psi|\dif\psi\rangle.
\end{equation}
The Uhlmann parallel-transport condition $W^\dagger \dot W = \dot W^\dagger
W$ becomes
$\langle\psi|\dot\psi\rangle = \langle\dot\psi|\psi\rangle$.
Since $\langle\psi|\psi\rangle=1$, this reduces to
the standard Berry condition $\langle\psi|\dot\psi\rangle=0$. The holonomy
is
then
\begin{equation}
\mathcal{U}_{\text{U}}(C)
=\exp\left(-\oint_C \langle\psi|\dif\psi\rangle\right)
=\me^{\mi\gamma_{\text{B}}},
\end{equation}
which is precisely the Berry phase factor. The Uhlmann phase
$\theta_{\text{U}}=\arg\tr_1[r(0)\mathcal{U}_{\text{U}}(C)]$ now reduces to
$\theta_{\text{U}}=\arg\me^{\mi\gamma_{\text{B}}}=\gamma_{\text{B}}$. This
establishes the complete
reduction to Berry geometry.

\subsubsection{Equal-weight projectors: Wilczek--Zee phase}

Consider an $m$-dimensional subspace with equal weights, i.e.,
$\rho=\frac{1}{m}P$ with $P$ a rank-$m$ projector ($P^2=P$). Then
$\lambda_a=1/m$ for all $a=1,\dots,m$, and the coefficient in
Eq.~\eqref{AU_explicit} becomes
$\frac{2\sqrt{\lambda_a\lambda_b}}{\lambda_a+\lambda_b}=1$. Thus
\begin{equation}
\mathcal{A}_{\text{U}}=E^\dagger \dif E=A_{\mathrm{WZ}},
\end{equation}
which is precisely the Wilczek--Zee non-Abelian Berry connection for the
$m$-fold degenerate subspace. The holonomy $\mathcal{U}_{\text{U}}(C)$ becomes
the Wilczek--Zee unitary matrix in $U(m)$.

We emphasize that this equality is exact within the fixed-rank stratum, in
contrast to the $T\to0$ correspondence between the full-rank Uhlmann phase
and the scalar Wilczek--Zee phase~\cite{bdpw-856t}, which is only
conditional: the zero-temperature limit of the full-rank connection retains a
coupling between different energy levels, whose holonomy carries a
$\pi_1(U(N))$ winding, whereas the minimal $U(k)$ connection here acts only
on the $k$-dimensional support and carries no kernel index, so no such
obstruction arises.

\subsubsection{Full-rank limit: standard Uhlmann connection}

We now verify that for $k=N$ our fixed-rank formula reduces to the standard
Uhlmann connection in the canonical gauge. Let $\rho$ be full rank and choose
an $N\times N$ unitary eigenframe $E$. In our gauge the section is
$s=E\sqrt{\Lambda}$, where $\Lambda=\diag(\lambda_1,\dots,\lambda_N)$. The
standard Uhlmann connection $\mathcal{A}_{\text{U}}^{(\sqrt\rho)}$ in the
canonical gauge $\sigma=\sqrt{\rho}=E\sqrt{\Lambda}E^\dagger$
satisfies~\cite{Guo20}
\begin{equation}
\rho \mathcal{A}_{\text{U}}^{(\sqrt\rho)}
+ \mathcal{A}_{\text{U}}^{(\sqrt\rho)}\rho
= [\sqrt{\rho},\dif\sqrt{\rho}],
\label{canonical_eq}
\end{equation}
or, in components,
\begin{equation}
(\mathcal{A}_{\text{U}}^{(\sqrt\rho)})_{ab}
= -\frac{(\sqrt{\lambda_a}-\sqrt{\lambda_b})^2}{\lambda_a+\lambda_b}
\langle a|\dif b\rangle,
\end{equation}
which follows directly from Eq.~\eqref{eq:full-rank-AU}.
The connection in our gauge, $\mathcal{A}_{\text{U}}^{(E)}$, is related to
the
canonical one by the gauge transformation generated by $E$:
\begin{equation}
\mathcal{A}_{\text{U}}^{(E)}
= E^\dagger \mathcal{A}_{\text{U}}^{(\sqrt\rho)} E + E^\dagger \dif E.
\end{equation}
Inserting the component expressions, we obtain
\begin{equation}
(\mathcal{A}_{\text{U}}^{(E)})_{ab}
= -\frac{(\sqrt{\lambda_a}-\sqrt{\lambda_b})^2}{\lambda_a+\lambda_b}
\langle a|\dif b\rangle
+ \langle a|\dif b\rangle.
\end{equation}
The sum simplifies as $-\frac{(\sqrt{\lambda_a}-\sqrt{\lambda_b})^2}{\lambda_a+\lambda_b} + 1
= \frac{2\sqrt{\lambda_a\lambda_b}}{\lambda_a+\lambda_b}$.
Hence
\begin{equation}
(\mathcal{A}_{\text{U}}^{(E)})_{ab}
= \frac{2\sqrt{\lambda_a\lambda_b}}{\lambda_a+\lambda_b}
\langle a|\dif b\rangle,
\end{equation}
which is exactly Eq.~\eqref{AU_explicit}. Therefore the fixed-rank theory
with
$k=N$ is gauge-equivalent to the standard Uhlmann theory.

These three limits establish the unified hierarchy:
\begin{equation}
\begin{array}{ccc}
k=1 & \rightarrow & U(1)\ \text{Berry geometry},\\[2pt]
1<k<N & \rightarrow & U(k)\ \text{fixed-rank Uhlmann geometry},\\[2pt]
\substack{\text{$m$-fold degenerate}\\\text{subspace}} & \rightarrow & U(m)\ \text{Wilczek--Zee connection},\\[2pt]
k=N & \rightarrow & U(N)\ \text{faithful Uhlmann geometry}.
\end{array}
\end{equation}

\section{Topology of the fixed-rank Uhlmann bundle}
\label{topology}

For $k=N$, the base manifold of faithful states is contractible and the
Uhlmann bundle is topologically trivial. For fixed rank $k<N$, the base
manifold is homotopy equivalent to the Grassmannian and the bundle can
acquire
nontrivial characteristic classes. In this section we analyze this topology
and its physical implications.

\subsection{Homotopy equivalence to the Grassmannian}

A key simplification is that $\mathcal{D}_k^N$ deformation-retracts onto the
submanifold of equal-weight projectors, which is diffeomorphic to
$\operatorname{Gr}(k,N)$. Let $\rho=ErE^\dagger$ with $E$ an orthonormal
frame
for $\operatorname{supp}\rho$ and $r=\diag(\lambda_1,\dots,\lambda_k)>0$.
Define
\begin{equation}
\rho_t=E\left[(1-t)r+\frac{t}{k}I_k\right]E^\dagger,\qquad 0\le t\le1.
\end{equation}
For all $t$, $\rho_t$ is positive semidefinite with unit trace and rank
exactly $k$; at $t=1$ it equals $\frac{1}{k}EE^\dagger$. Hence
\begin{equation}
\mathcal{D}_k^N\simeq \operatorname{Gr}(k,N),
\label{homotopy}
\end{equation}
so the topology of the fixed-rank Uhlmann bundle is controlled by the
tautological bundle over the Grassmannian. Physically, this means that all
nontrivial topology of the full fixed-rank bundle originates from the
variation of the supporting subspace, i.e., from the occupied-band geometry,
and not from the spectral weights.

\subsection{Tautological bundle and higher Chern classes}

On the equal-weight submanifold $\rho=\frac{1}{k}P$, the minimal
purification
satisfies $WW^\dagger=\frac{1}{k}P$, so $W=\frac{1}{\sqrt{k}}E\tilde U$ with
$\tilde U\in U(k)$. The bundle \eqref{principal_bundle} restricts to the
Stiefel bundle
\begin{equation}
U(k)\hookrightarrow V_k(\mathbb{C}^N)\longrightarrow\operatorname{Gr}(k,N),
\end{equation}
the unitary frame bundle of the tautological bundle
$\mathcal{S}_k\to\operatorname{Gr}(k,N)$, which is nontrivial for $0<k<N$.

From Eq.~\eqref{AU_explicit}, $(\mathcal{A}_{\text{U}})_{aa}
=\langle a|\dif a\rangle$, so summing over $a$ gives
$\tr \mathcal{A}_{\text{U}}=\tr K$ with $K=E^\dagger\dif E$. Taking the
exterior derivative,
\begin{equation}
\tr \mathcal{F}_{\text{U}}
=\dif \tr \mathcal{A}_{\text{U}}
=\tr(\dif K+K\wedge K)=\tr \mathcal{F}_K,
\end{equation}
and therefore
\begin{equation}
c_1(\text{Uhlmann bundle})=c_1(\mathcal{S}_k).
\end{equation}
For $k=1$, $\mathcal{S}_1=\mathcal{O}(-1)$ and $c_1=-1$, reproducing the
Hopf-bundle topology.

For $k\ge2$, higher Chern classes may appear. The total Chern class is
$c(\mathcal{S}_k)=1+c_1(\mathcal{S}_k)+\dots+c_k(\mathcal{S}_k)$, and all
characteristic classes of $\mathcal{S}_k$ are inherited by the Uhlmann
bundle. In particular, the second Chern number
\begin{equation}
C_2=\frac{1}{8\pi^2}\int_{M}\bigl[\tr(\mathcal{F}_{\text{U}}\wedge \mathcal{F}_{\text{U}})-\tr \mathcal{F}_{\text{U}}\wedge\tr \mathcal{F}_{\text{U}}\bigr]
\end{equation}
can be nonzero on a four-dimensional parameter manifold $M$, as shown below;
whenever $\tr\mathcal{F}_{\text{U}}=0$, as in the Yang monopole below, it
reduces to $\frac{1}{8\pi^2}\mathlarger{\int}_M\tr(\mathcal{F}_{\text{U}}\wedge\mathcal{F}_{\text{U}})$.
When the base manifold is a Brillouin zone torus and the state is an
equal-weight occupied projector, the Uhlmann Wilson loop along one cycle
reduces to the non-Abelian Berry--Wilczek--Zee Wilson loop, whose winding
reproduces $c_1$ and hence the quantized Hall response; when the parameter
space is four-dimensional, $C_2$ instead controls the quantized nonlinear
response of a four-dimensional quantum Hall system.

\subsection{Spectral degeneracy and the factorization of Chern classes}
\label{factorization}

\subsubsection{Factorization for non-degenerate spectra}

For non-degenerate spectrum, the off-diagonal factors
$2\sqrt{\lambda_a\lambda_b}/(\lambda_a+\lambda_b)$ are strictly less than
one,
so the Uhlmann connection differs from the Wilczek--Zee connection; its
characteristic classes, however, obey a rigid constraint.

\begin{proposition}\label{prop:splitting}
Let $\rho(p)$, $p\in M$, be a smooth family in $\mathcal{D}_k^N$ with
everywhere non-degenerate nonzero eigenvalues
$\lambda_1(p)>\dots>\lambda_k(p)>0$. Then the pullback of the Uhlmann bundle
to $M$ splits globally into eigenline bundles
\begin{equation}
L_a=\bigl\{(p,v)\;\bigl|\;v\in\ker(\rho(p)-\lambda_a(p))\bigr\},
\quad E=\bigoplus_{a=1}^k L_a,
\end{equation}
and
\begin{equation}
c(E)=\prod_{a=1}^k\bigl(1+c_1(L_a)\bigr),
\;
c_2(E)=\sum_{1\le a<b\le k}c_1(L_a)\,c_1(L_b).
\end{equation}
\end{proposition}

\begin{proof}
Since the eigenvalues are distinct, each spectral projector is obtained by
choosing locally a contour inside the spectral gap (for compact $M$ a uniform
gap permits a single fixed $\varepsilon$) in the contour integral
\begin{equation}
P_a(p)=\frac{1}{2\pi\mi}\oint\frac{\dif z}{z-\rho(p)};
\end{equation}
the resulting local spectral projectors agree on overlaps and therefore define
global smooth line subbundles $L_a$ spanning the support at every point. The
Chern class formula is the Whitney sum formula~\cite{Nakahara}.
\end{proof}

This is a literal realization of the splitting principle: the formal Chern
roots are the Berry Chern classes of the eigenstates~\cite{Nakahara,Kato}.
Hence with non-degenerate spectrum the higher Chern classes carry no
information beyond the Berry phases.
In particular, if
$H^2(M,\mathbb{Z})=0$, as for $M=S^4$, then $c_1(L_a)=0$ and $c_2(E)=0$: a
non-factorizable second Chern number \emph{requires} failure of the global
eigenline splitting, which in the present spectral setting requires
degeneracy. This sharpens the physical picture above: nontrivial higher Chern
topology is a property of the degenerate occupied subspace, not of the
spectral weights.

\emph{Example: the simplest non-degenerate family with nonzero $C_2$.}
Let $M=S^2_p\times S^2_q$ and $\mathcal{H}=\mathbb{C}^4
=\mathbb{C}^2\oplus\mathbb{C}^2$, with
\begin{equation}
|e_1(p,q)\rangle=\bigl(u(p),\,0\bigr),\quad
|e_2(p,q)\rangle=\bigl(0,\,v(q)\bigr),
\end{equation}
where $u(p),v(q)$ are Bloch spinors on the two spheres. Set
\begin{equation}
\rho(p,q)=\lambda_1|e_1\rangle\langle e_1|
+\lambda_2|e_2\rangle\langle e_2|
\end{equation}
with $\lambda_1>\lambda_2>0$ and $\lambda_1+\lambda_2=1$, which is
a rank-2 state with everywhere non-degenerate spectrum. Since
$\langle e_1|\dif e_2\rangle=0$, Eq.~\eqref{AU_explicit} gives
\begin{align}
\mathcal{A}_{\text{U}}
=\diag\!\left[
\frac{\mi}{2}(1-\cos\theta_p)\dif\phi_p,\;
\frac{\mi}{2}(1-\cos\theta_q)\dif\phi_q
\right],
\notag\\
\mathcal{F}_{\text{U}}
=\diag\!\left[
\frac{\mi}{2}\sin\theta_p\,\dif\theta_p\wedge\dif\phi_p,\;
\frac{\mi}{2}\sin\theta_q\,\dif\theta_q\wedge\dif\phi_q
\right],
\end{align}
independent of the weights. With $c_1(L_1)=-1$ on $S^2_p$ and
$c_1(L_2)=-1$ on $S^2_q$, so $c_1(E)\neq0$, the full Chern--Weil expression
gives
\begin{align}
C_2
&=\frac{1}{8\pi^2}\int_{S^2\times S^2}
\bigl[\tr(\mathcal{F}_{\text{U}}\wedge \mathcal{F}_{\text{U}})
-\tr \mathcal{F}_{\text{U}}\wedge\tr \mathcal{F}_{\text{U}}\bigr]\notag\\
&=c_1(L_1)\,c_1(L_2)=1,
\end{align}
saturating Proposition~\ref{prop:splitting}. The integration details are in
Appendix~\ref{app:s2xs2}.

\subsubsection{The full-rank case $k=N$}

The proposition applies verbatim at $k=N$, where its content changes
character. The full-rank bundle is trivial, since
$\sigma(\rho)=\sqrt{\rho}$ is a global section; every pullback is therefore
trivial and $c(E)=1$. Writing $x_a=c_1(L_a)$, the factorization formula now
requires all elementary symmetric polynomials of the eigenline classes to
vanish,
\begin{align}
&e_1=\sum_{a=1}^N x_a=0,\quad
e_2=\sum_{1\le a<b\le N}x_a x_b=0,\quad\dots,\notag\\
&e_N=x_1\cdots x_N=0,
\end{align}
so that the formula ceases to be a source of topology and becomes a set of
cancellation constraints on the spectral data. The individual $x_a$ may
still carry pure-state Berry topology; only their total package must
trivialize.

\emph{Example: a full-rank family with nontrivial eigenlines.} The family
\begin{equation}
\rho(\mathbf{n})=\frac12\bigl(I+r\,\hat{\mathbf{n}}\cdot\bm\sigma\bigr),
\qquad M=S^2,\quad 0<r<1,
\end{equation}
has non-degenerate eigenvalues $(1\pm r)/2$ and eigenlines with
$c_1(L_\pm)=\mp1$, while
\begin{equation}
E=L_+\oplus L_-=S^2\times\mathbb{C}^2
\end{equation}
is trivial. Accordingly $e_1=x_++x_-=0$, and $c_2(E)=x_+\smile x_-=0$:
here $\smile$ is the cup product that multiplies two degree-two classes
into a degree-four class, and it vanishes identically because
$H^4(S^2,\mathbb{Z})=0$, in agreement with $c(E)=1$. The formal product
of the two Chern numbers, $-1$, is not $c_2(E)$ and carries no cohomological
meaning here. On a four-dimensional base the constraints bite further:
eigenlines $L$ and $L^{-1}$ with $x^2\neq0$ would give
$c_2(E)=-x^2\neq0$, contradicting the triviality of $E$, which at $k=N$ is
the pullback of the trivial full-rank Uhlmann bundle; such a spectral
geometry is therefore unattainable.

Two conclusions follow for the fixed-rank theory. First, the novelty of the
non-degenerate Uhlmann connection lies in its weight-dependent
\emph{holonomy matrices}, not in its characteristic classes, which are
already fixed by the eigenlines. Second, non-factorizable higher topology
requires spectral degeneracy; the equal-weight Yang monopole of
Sec.~\ref{yang} is the minimal realization, where
$H^2(S^4,\mathbb{Z})=0$ excludes any factorization and the topology is
detected solely by $C_2=\pm1$. This is also the case of direct physical
interest, since $C_2$ governs the quantized response of a four-dimensional
quantum Hall system and can be accessed through the non-Abelian Wilson loop
of the degenerate occupied subspace.

\subsection{Geometric versus topological character of the Uhlmann phase}
\label{phase_vs_topology}

We finally clarify the topological status of the scalar Uhlmann phase.
In general $\theta_{\text{U}}=\arg\tr[r(0)\mathcal{U}_{\text{U}}(C)]$ is
gauge invariant but geometric: it varies continuously with the loop and is
therefore in general not quantized. The underlying holonomy
$\mathcal{U}_{\text{U}}(C)$, by contrast, becomes
homotopy invariant precisely when the connection is flat,
$\mathcal{F}_{\text{U}}\equiv0$: homotopic loops $C_1,C_2$ then bound a
surface $\Sigma$, and the non-Abelian Stokes theorem gives
$\mathcal{U}_{\text{U}}(C_1)\mathcal{U}_{\text{U}}(C_2)^{-1}
=\mathcal{P}\exp(-\mathlarger{\int}_\Sigma\mathcal{F}_{\text{U}})=I_k$;
conversely, a point of nonzero curvature produces a small loop with
nonidentity holonomy. On the simply connected base
$\mathcal{D}_k^N\simeq\operatorname{Gr}(k,N)$, a flat connection would
render all holonomies trivial, so a nonvanishing Uhlmann phase necessarily
requires nonvanishing curvature and is generically geometric rather than
topological. The topological content of the theory resides instead in the
characteristic classes of the bundle, as the Chern data of
Secs.~\ref{example_phase} and \ref{yang} demonstrate.

Quantization is nevertheless not foreign to the Uhlmann phase. In much of
the full-rank literature, $\theta_{\text{U}}$, restricted to particular
loops and symmetries, takes only the values $0$ and $\pi$, and its jump as a
control parameter crosses a critical value has been studied under the name
of a finite-temperature topological phase
transition~\cite{Viyuela14,UP2D15}, with the sign reversal of the overlap
interpreted as the parallel-transported purification switching from a
M\"obius-strip-like, antiparallel return to a cylinder-like, parallel one.
The terminology is operationally motivated: $\theta_{\text{U}}$ is a direct
function of the Uhlmann holonomy, the mixed-state descendant of the Berry
holonomy, and its quantized value below the transition is inherited from the
zero-temperature topology, the jump thus diagnosing the destruction of a
topological signature. In the strict sense, however, no topological
invariant of the bundle changes at the transition, which is instead a node
of the overlap amplitude; we therefore describe such events as geometric
phase transitions and reserve the word topology for the characteristic
classes.

Even within the geometric domain, however, the phase may still reflect part
of the bundle's topology, and the question is how much. This is the Berry
analogue: can $\theta_{\text{U}}$ be proportional to the first Chern number?
The bridge is
$\tr\mathcal{F}_{\text{U}}=\tr \mathcal{F}_K$, i.e.,
$c_1(\text{Uhlmann bundle})=c_1(\mathcal{S}_k)$. For $k=1$ the answer is
affirmative: $\theta_{\text{U}}$ reduces to the Berry phase
$\gamma_{\text{B}}=\mathlarger{\int}_\Sigma F$, which, for the rotationally
symmetric representatives and equatorial loops considered below, equals
$\pi c_1$ modulo $2\pi$; for a generic loop or curvature distribution the
Berry phase is a flux through a spanning surface and is not fixed by $c_1$
alone. For $k>1$ with an Abelian connection within $U(k)$, as in the
rank-2 example of Sec.~\ref{example_phase}, the same statement survives,
again for the symmetry-adapted loop of that construction, in
$\mathbb{Z}_2$ form: $\theta_{\text{U}}$ is restricted to $\{0,\pi\}$,
reflecting the parity of the twisted eigenline's Berry holonomy. For
genuinely non-Abelian families with $k\ge2$, however, $\theta_{\text{U}}$ is
a highly compressed gauge-invariant functional of the full holonomy: it may
vary with the non-Abelian curvature, as the Yang monopole below illustrates,
but it retains too little information to reconstruct the holonomy matrix or
the higher characteristic classes.

Beyond the first Chern class, higher invariants such as the second Chern
number of the Yang monopole are likewise invisible to the scalar phase; the
non-Abelian Wilson loop is their natural probe, its determinant's winding
reproducing $c_1$ in two dimensions, while a precise extraction of $C_2$
from Wilson-loop data lies beyond the scope of this work. Within the
symmetry-adapted constructions above, the scalar phase thus carries no more
than a $\mathbb{Z}_2$ shadow of the first Chern class; the integer topology
of the fixed-rank bundle lives in its characteristic classes and Wilson
loops.

\section{Examples}
\label{examples}

To make the abstract framework concrete, we study several examples
illustrating different facets of the fixed-rank Uhlmann geometry. Each
highlights a distinct feature: spectral dependence of an Abelian phase;
non-Abelian holonomy from a moving support with unequal eigenvalues; a
solvable dissipative orbit; and a rank-2 Yang monopole with quantized second
Chern number. Together they show that the fixed-rank theory contains the
known pure-state and equal-weight limits while admitting new mixed-state
geometric and topological structures.

\subsection{Uhlmann phase for a rank-2 state on $\mathbb{C}P^1$}
\label{example_phase}

We begin with an elementary rank-2 example that illustrates the computation
of an Abelian Uhlmann phase. Let $N=3$, $k=2$, and choose two orthonormal
vectors in $\mathbb{C}^3$,
\begin{equation}
|e_1(\theta,\phi)\rangle
=\begin{pmatrix}
\cos\frac{\theta}{2}\\
\me^{\mi\phi}\sin\frac{\theta}{2}\\
0
\end{pmatrix},
\qquad
|e_2\rangle
=\begin{pmatrix}
0\\0\\1
\end{pmatrix}.
\end{equation}
Define $\rho(\theta,\phi)=\lambda_1|e_1\rangle\langle e_1|
+\lambda_2|e_2\rangle\langle e_2|$ with $\lambda_1,\lambda_2>0$ and
$\lambda_1+\lambda_2=1$. The support is a two-dimensional subspace varying
with $(\theta,\phi)$ as a $\mathbb{C}P^1$ embedded in
$\operatorname{Gr}(2,3)$.

Since $|e_2\rangle$ is constant, only the diagonal matrix element
$(\mathcal{A}_{\text{U}})_{11}=\langle e_1|\dif e_1\rangle
=\frac{\mi}{2}(1-\cos\theta)\dif\phi$ is nonzero. Thus
\begin{equation}
\mathcal{A}_{\text{U}}
=\frac{\mi}{2}(1-\cos\theta)\dif\phi
\begin{pmatrix}1&0\\0&0\end{pmatrix},
\end{equation}
which is Abelian inside $U(2)$ but produces a genuine mixed-state holonomy.
For the loop $\theta=\pi/2$, $\phi\in[0,2\pi]$, the holonomy is diagonal,
\begin{align}
\mathcal{U}_{\text{U}}(C)&=\exp\left(-\oint_C \mathcal{A}_{\text{U}}\right)=\exp\left[-\mi\pi \begin{pmatrix}1&0\\0&0\end{pmatrix}\right]\notag\\
&=\diag(\me^{-\mi\pi},1)=\diag(-1,1),
\end{align}
and the Uhlmann phase $\theta_{\text{U}}=\arg\tr_2[r(0)\mathcal{U}_{\text{U}}(C)]$
with $r(0)=\diag(\lambda_1,\lambda_2)$ gives
\begin{equation}
\theta_{\text{U}}=\arg(\lambda_2-\lambda_1)
=\begin{cases}
0, & \lambda_2>\lambda_1,\\
\pi, & \lambda_1>\lambda_2.
\end{cases}
\end{equation}
At $\lambda_1=\lambda_2$ the amplitude vanishes and the phase jumps by
$\pi$, signaling a phase transition driven purely by spectral
degeneracy. The resulting phase is plotted in Fig.~\ref{fig:Ea}.

\begin{figure}[ht]
\centering
\includegraphics[width=3.2in]{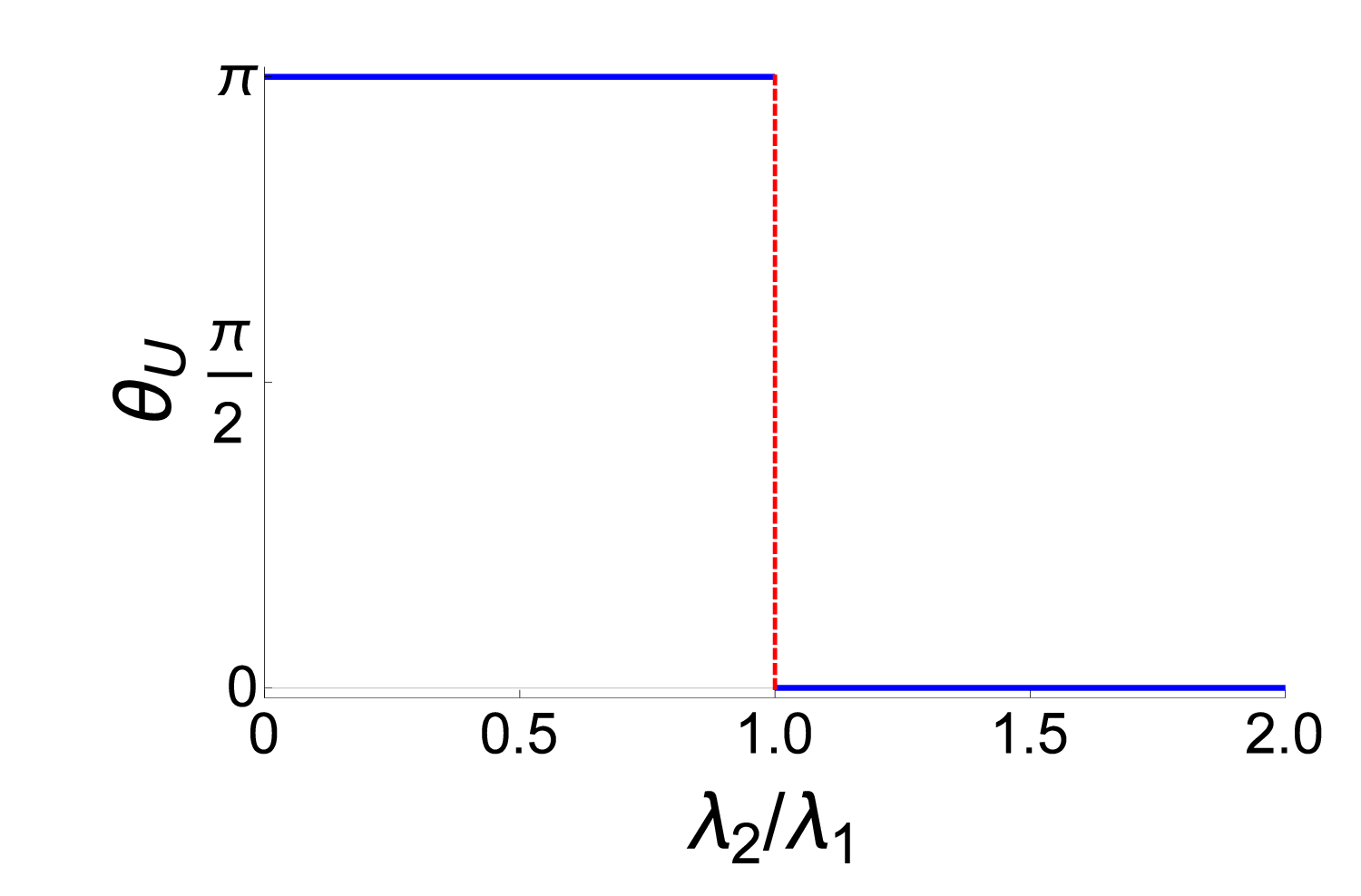}
\caption{Uhlmann phase $\theta_{\text{U}}$ as a function of the ratio
$\lambda_2/\lambda_1$ for the rank-2 state on $\mathbb{C}P^1$. The phase
jumps from $\pi$ to $0$ at $\lambda_1=\lambda_2$.}
\label{fig:Ea}
\end{figure}

Although the result appears elementary, it is instructive to ask what
geometric or topological information $\theta_{\text{U}}$ actually carries.
Topologically, $L_1=\operatorname{span}|e_1(\theta,\phi)\rangle$ is the
tautological line bundle over $\mathbb{C}P^1$ with $c_1(L_1)=-1$, while
$L_2=\operatorname{span}|e_2\rangle$ is trivial with $c_1(L_2)=0$. For the
equatorial loop of this rotationally symmetric representative, which encloses
one half of the total Chern flux, the Berry phase of each eigenline equals
\begin{equation}
\phi_{1,2} = \pi c_1(L_{1,2}) \pmod{2\pi},
\end{equation} a relation special to this geometry that will be used below.
The Uhlmann amplitude is
$\mathcal{Z}_{\text{U}}=\tr_2[r(0)\mathcal{U}_{\text{U}}(C)]$, and since
$\mathcal{U}_{\text{U}}(C)=\diag(\me^{\mi\phi_1},\me^{\mi\phi_2})$ with
$\me^{\mi\phi_a}=\me^{\mi\pi c_1(L_a)}=(-1)^{c_1(L_a)}$, it evaluates to
\begin{equation}
\mathcal{Z}_{\text{U}}
=(-1)^{c_1(L_1)}\lambda_1+(-1)^{c_1(L_2)}\lambda_2.
\end{equation}
Depending on the Chern parities, three cases arise for such
symmetry-adapted loops, where $c_1^{L_a}\equiv c_1(L_a)$:
\begin{center}
\small
\begin{tabular}{c c c}
\hline
$(c_1^{L_1},c_1^{L_2})\bmod 2$ & $\mathcal{Z}_{\text{U}}$ & $\theta_{\text{U}}$ \\
\hline
$(0,0)$ & $\lambda_1+\lambda_2=1>0$ & $0$, independent \\
$(1,1)$ & $-\lambda_1-\lambda_2=-1<0$ & $\pi$, independent \\
$(1,0)$ or $(0,1)$ & $\pm(\lambda_2-\lambda_1)$ & $0$ or $\pi$, by weights \\
\hline
\end{tabular}
\end{center}
Clearly, $\mathcal{Z}_{\text{U}}$ in the first two cases is independent of
the spectral weights. The present example belongs to the third case:
$c_1(L_1)=-1$ and $c_1(L_2)=0$, so
$\mathcal{Z}_{\text{U}}=\lambda_2-\lambda_1$, and the spectral weights
choose
between $0$ and $\pi$. Within this construction, $\theta_{\text{U}}$ does
not
see the integer value of $c_1$; it sees only the parity of the Berry
holonomy, and even that parity does not fix the phase uniquely when the two
parities differ. The Uhlmann phase is therefore not a topological invariant
here; it is a $\mathbb{Z}_2$-valued shadow of the first Chern class specific
to this symmetric setting, in contrast to the Berry phase itself, which for
the equatorial loop equals $\pi c_1$ modulo $2\pi$ and is independent of the
spectral weights. This is the simplest instance showing that the rank-2
Uhlmann phase can be nontrivial even for Abelian holonomy.

\begin{figure*}[ht]
\centering
\includegraphics[width=0.92\linewidth]{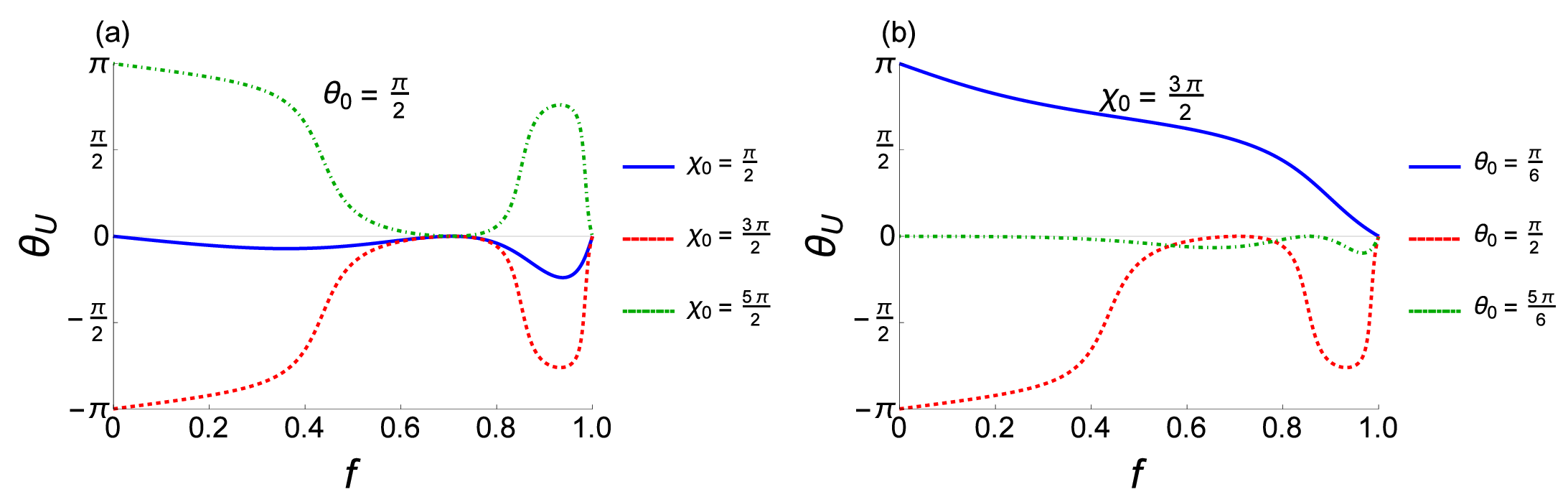}
\caption{The Uhlmann phase $\theta_{\text{U}}$  as a
function of the spectral factor $f$.
(a) Fixed $\theta_0=\pi/2$ and three values of $\chi_0$; (b) fixed
$\chi_0=3\pi/2$ and three values of $\theta_0$. The phase varies
continuously and non-monotonically with $f$, reaching values of order $\pi$
at intermediate weights, and vanishes at $f=1$ where
$\lambda_1=\lambda_2=\frac{1}{2}$.
}
\label{fig:Eb}
\end{figure*}
\subsection{Non-Abelian Uhlmann holonomy in a rank-2 $N=3$ model}
\label{nonabelian_example}

We now construct a model with genuinely non-Abelian Uhlmann holonomy and
unequal eigenvalues.

\subsubsection{Moving support and eigenframe}

Using the Gell-Mann matrices\begin{equation}
\lambda_5=
\begin{pmatrix}
0&0&-\mi\\
0&0&0\\
\mi&0&0
\end{pmatrix},
\qquad
\lambda_7=
\begin{pmatrix}
0&0&0\\
0&0&-\mi\\
0&\mi&0
\end{pmatrix},
\end{equation}and the diagonal generator
$\Lambda=\frac{1}{3}\diag(-1,-1,2)$, define the $SU(3)$ transformation
\begin{equation}
U(\phi,\chi,\theta)
=\me^{-\mi\phi\lambda_5/2}\me^{-\mi\chi\Lambda}\me^{-\mi\theta\lambda_7/2}.
\end{equation}
From the reference frame $E_0=(|1\rangle,|2\rangle)=\begin{pmatrix}
1&0\\
0&1\\
0&0
\end{pmatrix}$ we obtain the moving
eigenframe $E(\phi,\chi,\theta)=U(\phi,\chi,\theta)E_0$, whose orthonormal
columns are
\begin{equation}
|e_1\rangle
=\begin{pmatrix}
\me^{\mi\chi/3}\cos\frac{\phi}{2}\\
0\\
\me^{\mi\chi/3}\sin\frac{\phi}{2}
\end{pmatrix},
\;
|e_2\rangle
=\begin{pmatrix}
-\me^{-2\mi\chi/3}\sin\frac{\phi}{2}\sin\frac{\theta}{2}\\
\me^{\mi\chi/3}\cos\frac{\theta}{2}\\
\me^{-2\mi\chi/3}\cos\frac{\phi}{2}\sin\frac{\theta}{2}
\end{pmatrix}.
\end{equation}

\subsubsection{Frame connection and Uhlmann connection}

The frame connection $K=E^\dagger\dif E$ has components
\begin{align}
&K_\phi
=\frac{\sin\frac{\theta}{2}}{2}
\begin{pmatrix}0&-\me^{-\mi\chi}\\
\me^{\mi\chi}&0\end{pmatrix},
\quad
K_\theta=0,\notag\\
\qquad
&K_\chi
=\begin{pmatrix}
\frac{\mi}{3}&0\\
0&-\frac{\mi}{6}(1-3\cos\theta)
\end{pmatrix}.
\end{align}
In terms of Pauli matrices,
\begin{align}
K
&=\frac{\mi}{2}\sin\frac{\theta}{2}
\left(\sin\chi\,\sigma_x-\cos\chi\,\sigma_y\right)\dif\phi\notag\\
&\quad+\frac{\mi}{12}(1+3\cos\theta)I_2\,\dif\chi
+\frac{\mi}{4}(1-\cos\theta)\sigma_z\,\dif\chi.
\end{align}
For a rank-2 state
$\rho=\lambda_1|e_1\rangle\langle e_1|
+\lambda_2|e_2\rangle\langle e_2|=ErE^\dagger$ with
$\lambda_1\neq\lambda_2$ and $r=\diag(\lambda_1,\lambda_2)$,
Eq.~\eqref{AU_explicit} multiplies the off-diagonal
elements of $K$ by
$f=2\sqrt{\lambda_1\lambda_2}$ while leaving the diagonal ones unchanged:
\begin{align}
\mathcal{A}_{\text{U}}
&=\frac{\mi f}{2}\sin\frac{\theta}{2}
\left(\sin\chi\,\sigma_x-\cos\chi\,\sigma_y\right)\dif\phi
\notag\\&\quad+\frac{\mi}{12}(1+3\cos\theta)I_2\,\dif\chi
+\frac{\mi}{4}(1-\cos\theta)\sigma_z\,\dif\chi.
\end{align}
Since $[A_\phi,A_\chi]\neq0$ whenever $f\neq0$ and $\theta\neq0$, this
is a genuinely non-Abelian Uhlmann connection. This is the first explicit
demonstration that the fixed-rank Uhlmann holonomy goes beyond the
Wilczek--Zee one: although the connection differs from $E^\dagger\dif E$
only by spectral factors on the interlevel matrix elements, path ordering of
the resulting transport produces genuinely non-Abelian holonomies with no
Wilczek--Zee counterpart.

\subsubsection{Closed loop and holonomy}

Fixing $\theta=\theta_0$ and taking the rectangular loop
$(0,0)\to(4\pi,0)\to(4\pi,\chi_0)\to(0,\chi_0)\to(0,0)$ in the
$(\phi,\chi)$ plane, with $\phi$ defined modulo $4\pi$, we define
\begin{equation}
a=\frac{f}{2}\sin\frac{\theta_0}{2},
\quad
b=\frac{1+3\cos\theta_0}{12},
\quad
d=\frac{1-\cos\theta_0}{4}.
\end{equation}
Along the four segments the connection is constant, so the path-ordered
exponential reduces to ordinary exponentials:
\begin{align}
&U_1=\me^{\mi 4\pi a\sigma_y},\quad
U_2=\me^{-\mi\chi_0(bI_2+d\sigma_z)},\notag\\
&U_3=\me^{\mi 4\pi a(\sin\chi_0\sigma_x-\cos\chi_0\sigma_y)},\quad
U_4=U_2^{-1}.
\end{align}
The total holonomy is $\mathcal{U}_{\text{U}}(C)=U_4U_3U_2U_1$, and the
$bI_2$ factor cancels between $U_2$ and $U_4$. Conjugating with
$\me^{\pm\mi\chi_0 d\sigma_z}$ yields
\begin{equation}
\mathcal{U}_{\text{U}}(C)
=\me^{\mi 4\pi a(\sin\alpha\,\sigma_x-\cos\alpha\,\sigma_y)}
\me^{\mi 4\pi a\sigma_y},
\end{equation}
where $\alpha=\chi_0\frac{1+\cos\theta_0}{2}$.
With $\gamma=4\pi a$, the trace is
\begin{equation}
\tr\mathcal{U}_{\text{U}}(C)
=2\left[\cos^2\gamma+\sin^2\gamma\cos\alpha\right].
\end{equation}
For generic $\theta_0,\chi_0$ and $0<f<1$, this is neither $2$ nor $-2$, so
the holonomy is a nontrivial $SU(2)$ matrix; its non-Abelian character
resides in the connection itself, through $[A_\phi,A_\chi]\neq0$.

The Uhlmann phase follows from Eq.~\eqref{theta_U} with
$r(0)=\diag(\lambda_1,\lambda_2)$ in the initial eigenframe. The holonomy
matrix itself reads
\begin{equation}
\mathcal{U}_{\text{U}}(C)=
\begin{pmatrix}
\cos^2\gamma+\sin^2\gamma\, \me^{-\mi\alpha}&
\sin\gamma\cos\gamma\,(1-\me^{-\mi\alpha})\\[2pt]
\sin\gamma\cos\gamma\,(\me^{\mi\alpha}-1)&
\cos^2\gamma+\sin^2\gamma\, \me^{\mi\alpha}
\end{pmatrix},
\end{equation}
and its weighted trace takes the closed form
\begin{align}
\mathcal{Z}_{\text{U}}
&=\tr[r(0)\mathcal{U}_{\text{U}}(C)]\notag\\
&=\cos^2\gamma+\sin^2\gamma\cos\alpha
+\mi(\lambda_2-\lambda_1)\sin^2\gamma\sin\alpha.
\label{Z_nonabelian}
\end{align}
For generic parameters, $\mathcal{Z}_{\text{U}}$ is complex: the Uhlmann
phase is not quantized and varies continuously with the spectral weights and
the loop. This continuous variation is illustrated in
Fig.~\ref{fig:Eb}, which plots
$\theta_{\text{U}}$ against $f$ for representative loop parameters: the
phase traces a non-monotonic curve of order $\pi$ at intermediate weights.
For the representative loops shown in Fig.~\ref{fig:Eb}, the phase
approaches zero as $f\to1$, where $\lambda_1=\lambda_2=\frac{1}{2}$; in this
equal-weight limit the spectral dressing disappears and the Uhlmann
connection reduces to the Wilczek--Zee form $E^\dagger\dif E$.

$\theta_{\text{U}}$ reduces to a $\mathbb{Z}_2$-valued phase only
in the degenerate limit $\lambda_1\to\lambda_2$, where
$\mathcal{Z}_{\text{U}}=\tfrac12\tr\mathcal{U}_{\text{U}}(C)$ becomes real,
or for fine-tuned loops with $\sin\alpha=0$. The weight dependence of the
connection, and hence of the holonomy, enters through the off-diagonal
factor $f$; the scalar amplitude $\mathcal{Z}_{\text{U}}$ carries an
additional dependence on the initial weights through
$r(0)=\diag(\lambda_1,\lambda_2)$.

Topologically, this family realizes no invariant. Both eigenlines admit
global nonvanishing sections, $|e_1\rangle$ and $|e_2\rangle$ themselves, so
$L_1$ and $L_2$ are trivial line bundles with $c_1(L_1)=c_1(L_2)=0$, and the
Uhlmann bundle over this family is likewise trivial. This holds despite the
nonvanishing curvature $\mathcal{F}_{22}=-\frac{\mi}{2}\sin\theta\,\dif\theta\wedge\dif\chi$,
which is an exact two-form because its connection one-form
$K_{22}=\frac{\mi}{6}(3\cos\theta-1)\dif\chi$ is globally defined. All
spectral information enters through the scalar factor $f$ and affects only
geometric data: it cannot generate or modify any topological invariant. This
example is therefore purely geometric: its content is the weight-dependent
non-Abelian holonomy of a topologically trivial bundle, complementing the
topological examples of Secs.~\ref{example_phase} and \ref{yang}.

\subsection{Dissipatively driven fixed-rank orbit}
\label{lindblad_example}

Our third example is a genuinely dissipative model: a reservoir exchanges
population \emph{within} the two-dimensional support while a Hamiltonian
rotates the support itself. The dissipator acts nontrivially for generic times throughout
the cycle, in contrast to a decoherence-free subspace, while the rank is
preserved exactly, so that the reservoir participates directly in the
Uhlmann transport.

\subsubsection{Model and exact rank-2 solution}

On $\mathcal H=\mathbb C^3$ with orthonormal basis
$\{|1\rangle,|2\rangle,|3\rangle\}$, take the spin-1 generator
\begin{equation}
J_x=\frac{1}{\sqrt2}
\begin{pmatrix}
0&1&0\\
1&0&1\\
0&1&0
\end{pmatrix},
\end{equation}
whose eigenvalues are $-1$, $0$, $1$, and define
\begin{equation}
U(\phi)=\me^{-\mi\phi J_x},\quad \phi=\omega t,\quad
|u_a(\phi)\rangle=U(\phi)|a\rangle .
\end{equation}
Since $\me^{-\mi2\pi J_x}=\mathbb I$, the loop traversed in one period
$T=2\pi/\omega$ is strictly closed. The Hamiltonian is $H=\omega J_x$, and
the reservoir is modeled by two periodically modulated jump operators
\begin{align}
L_+(t)=\sqrt{\gamma_+(t)}\,|u_1(\phi)\rangle\langle u_2(\phi)|,\notag\\
L_-(t)=\sqrt{\gamma_-(t)}\,|u_2(\phi)\rangle\langle u_1(\phi)|.
\end{align}
Here $\gamma_\pm(t)=\Gamma\left[\frac12\pm\delta\cos(\omega t)\right]$ with $0<\delta<\frac12$, 
which are manifestly non-negative. We seek a solution of the Lindblad
equation
\begin{equation}
\dot\rho=-\mi[H,\rho]
+\sum_{s=\pm}\Bigl(L_s\rho L_s^\dagger
-\tfrac12\{L_s^\dagger L_s,\rho\}\Bigr)
\label{lindblad}
\end{equation}
of the form
\begin{equation}
\rho(t)=U(\phi)\,\diag[p(t),1-p(t),0]\,U^\dagger(\phi).
\label{rho_ansatz}
\end{equation}
Because $[J_x,U]=0$, the Hamiltonian contribution to $\partial_t\rho$
reproduces the frame rotation exactly,
\begin{equation}
\partial_t\rho
=-\mi[H,\rho]
+U\,\diag(\dot p,-\dot p,0)\,U^\dagger,
\end{equation}
while the dissipators, evaluated with
$\langle u_1|\rho|u_1\rangle=p$ and $\langle u_2|\rho|u_2\rangle=1-p$, give
\begin{align}
\mathcal D[L_+]\rho
&=\gamma_+(1-p)\bigl(|u_1\rangle\langle u_1|-|u_2\rangle\langle u_2|\bigr),\notag\\
\mathcal D[L_-]\rho
&=-\gamma_-p\bigl(|u_1\rangle\langle u_1|-|u_2\rangle\langle u_2|\bigr).
\label{diss_eval}
\end{align}
The Lindblad equation \eqref{lindblad} is therefore equivalent to the
scalar population equation
\begin{equation}
\dot p=\gamma_+(1-p)-\gamma_-p
=-\Gamma\bigl[p-p_{\rm eq}(t)\bigr],
\label{pop_eq}
\end{equation}
where $p_{\rm eq}(t)=\frac12+\delta\cos\omega t$. 
Two comments are in order. First, Eq.~\eqref{diss_eval} shows that the
dissipator alone carries the entire spectral dynamics,
$\sum_s\mathcal D[L_s]\rho
=\dot p\,(|u_1\rangle\langle u_1|-|u_2\rangle\langle u_2|)$,
and vanishes only at the isolated instants where $\dot p=0$. Second, the
separation between geometry and transport is exact here: the Hamiltonian
rotates the support, and the reservoir dresses the weights.

The general solution of Eq.~\eqref{pop_eq} is
$p(t)=p_{\rm per}(t)+C\me^{-\Gamma t}$, and after the transient decay the
density matrix settles on the periodic steady state
\begin{equation}
p(t)=\frac12+A\cos(\omega t-\tilde\varphi),\;
A=\frac{\delta\Gamma}{\sqrt{\Gamma^2+\omega^2}},\;
\tan\tilde\varphi=\frac{\omega}{\Gamma}.
\label{p_per}
\end{equation}
For $\Gamma\gg\omega$ the system tracks the reservoir
($p\simeq p_{\rm eq}$), whereas for $\Gamma\ll\omega$ the rapid
modulation averages out in the periodic steady state and the weights collapse
to $p(t)\simeq1/2$. In either regime,
\begin{equation}
A\leq\delta<\frac12
\;\Longrightarrow\;
0<\frac12-A\leq p(t)\leq\frac12+A<1,
\end{equation}
so the nonzero eigenvalues $p(t)$ and $1-p(t)$ never touch $0$ or $1$:
the trajectory lives in $\mathcal D_2^3$ for all times, with rank exactly
$2$ rather than approaching rank $2$ only in a limiting sense. Note also that $p(t)$ crosses
$1/2$ twice per period, where the instantaneous spectrum is degenerate;
the construction remains smooth because the chosen smooth frame
$E=(|u_1\rangle,|u_2\rangle)$ remains globally defined by $U(\phi)$ and
$r(\phi)=\diag[p,1-p]$ stays positive, so the Sylvester equation
\eqref{AU_sylvester} never degenerates.

\subsubsection{Uhlmann connection and Wilson loop}

The frame connection of $E$ follows from
$U^\dagger\partial_\phi U=-\mi J_x$:
the diagonal matrix elements vanish,
$\langle u_a|\partial_\phi u_a\rangle=-\mi\langle a|J_x|a\rangle=0$
($a=1,2$), and
$\langle u_1|\partial_\phi u_2\rangle=-\mi\langle1|J_x|2\rangle=-\mi/\sqrt2$.
Since $\lambda_1+\lambda_2=1$, the dressing factor in
Eq.~\eqref{AU_explicit} reduces to $f(\phi)=2\sqrt{\lambda_1\lambda_2}$,
and
\begin{align}
\mathcal{A}_{\text{U}}
=-\frac{\mi}{\sqrt2}f(\phi)\,\sigma_x,
\label{AU_diss}
\end{align}
where $f(\phi)=\sqrt{1-4A^2\cos^2(\phi-\tilde\varphi)}$. 
All components of $\mathcal{A}_{\text{U}}$ are proportional to the same
generator $\sigma_x$, so
$[\mathcal{A}_{\text{U}}(\phi_1),\mathcal{A}_{\text{U}}(\phi_2)]=0$ and the
path ordering of the holonomy collapses. Hence
\begin{equation}
\mathcal{U}_{\text{U}}(C)
=\exp\left[\frac{\mi\sigma_x}{\sqrt2}
\int_0^{2\pi}f(\phi)\,\dif\phi\right],
\end{equation}
where the integral is a complete elliptic integral of the second kind. Introducing
\begin{equation}
E(m)=\int_0^{\frac{\pi}{2}}\sqrt{1-m\sin^2\theta}\dif\theta=\frac{1}{4}\int_0^{2\pi}f(\phi)\dif\phi,
\end{equation}
with $m=4A^2=\frac{4\delta^2\Gamma^2}{\Gamma^2+\omega^2}$, and defining
$\Theta\equiv 2\sqrt2\,E(m)$, we obtain
\begin{equation}
\mathcal{U}_{\text{U}}(C)=\me^{\mi\Theta\sigma_x}.
\end{equation}
The condition $\delta<1/2$ guarantees $m\leq4\delta^2<1$, so $E(m)$ is regular and $f$ never vanishes.
The natural gauge-invariant observable is the Wilson loop
\begin{equation}
W(C)=\tr\,\mathcal{U}_{\text{U}}(C)=2\cos\Theta
=2\cos\!\bigl[2\sqrt2\,E(m)\bigr],
\label{W_diss}
\end{equation}
which depends on the modulation depth $\delta$ and the dimensionless
reservoir-to-drive ratio $\Gamma/\omega$.

 The scalar Uhlmann phase carries less information here. Since
$r(0)=\diag[p_0,1-p_0]$ while
$\mathcal{U}_{\text{U}}(C)=\cos\Theta\,I_2+\mi\sin\Theta\,\sigma_x$,
the off-diagonal part of the holonomy drops out of the weighted trace,
\begin{equation}
\mathcal{Z}_{\text{U}}
=\tr\bigl[r(0)\mathcal{U}_{\text{U}}(C)\bigr]
=\cos\Theta ,
\label{Z_diss}
\end{equation}
independently of the initial weight $p_0$.
Moreover, the scalar phase is completely blind to the dissipative
dressing. Because $m\in[0,4\delta^2)$ with $\delta<1/2$ and $E(m)$ is
monotonically decreasing, $\Theta=2\sqrt2\,E(m)$ ranges over
$(2\sqrt2,\sqrt2\pi]\subset(\pi/2,3\pi/2)$, so $\cos\Theta<0$
throughout the reachable parameter range. Consequently
\begin{equation}
\mathcal{Z}_{\text{U}}=\cos\Theta<0,\qquad \theta_{\text{U}}=\pi
\end{equation}
identically: no node, no jump, and no undefined point occurs for any
$\Gamma/\omega$ or $\delta$. The continuous spectral information carried
by the holonomy is compressed entirely into the Wilson loop
$W(C)=2\cos\Theta$, the natural observable of this dissipative geometry.
This is the scalar-phase--Wilson-loop hierarchy of
Sec.~\ref{phase_vs_topology} in its extreme form: the scalar phase is
constant while the Wilson loop probes the full dissipative crossover.

Although the scalar phase is blind, the Wilson loop itself has one
distinctive feature. As $\Gamma/\omega$ increases, $\Theta$ decreases
through $\pi$ whenever $2\sqrt2\,E(4\delta^2)<\pi$, and $W(C)$ then
attains its minimum $-2$: the holonomy passes through
$\mathcal{U}_{\text{U}}(C)=-\mathbb{I}_2$, the nontrivial central element
of $SU(2)$, i.e., the parallel-transported purification returns to minus
itself, the M\"obius-like antiparallel return of
Sec.~\ref{phase_vs_topology}, reached here by tuning the reservoir rather
than a Hamiltonian parameter. For $\delta=0.49$ this occurs at
$(\Gamma/\omega)_c\simeq3.6$ (dot in Fig.~\ref{fig:Ec}); for weaker
modulation the crossing does not occur and $W(C)$ decreases monotonically.
The dot is not a phase transition: $W(C)$ is smooth through it, and the
scalar phase remains $\pi$ on both sides (indeed $\theta_{\text{U}}=\pi$
identically since $W(C)<0$ throughout; see Fig.~\ref{fig:Ec}(b)).

\begin{figure*}[ht]
\centering
\includegraphics[width=6.4in]{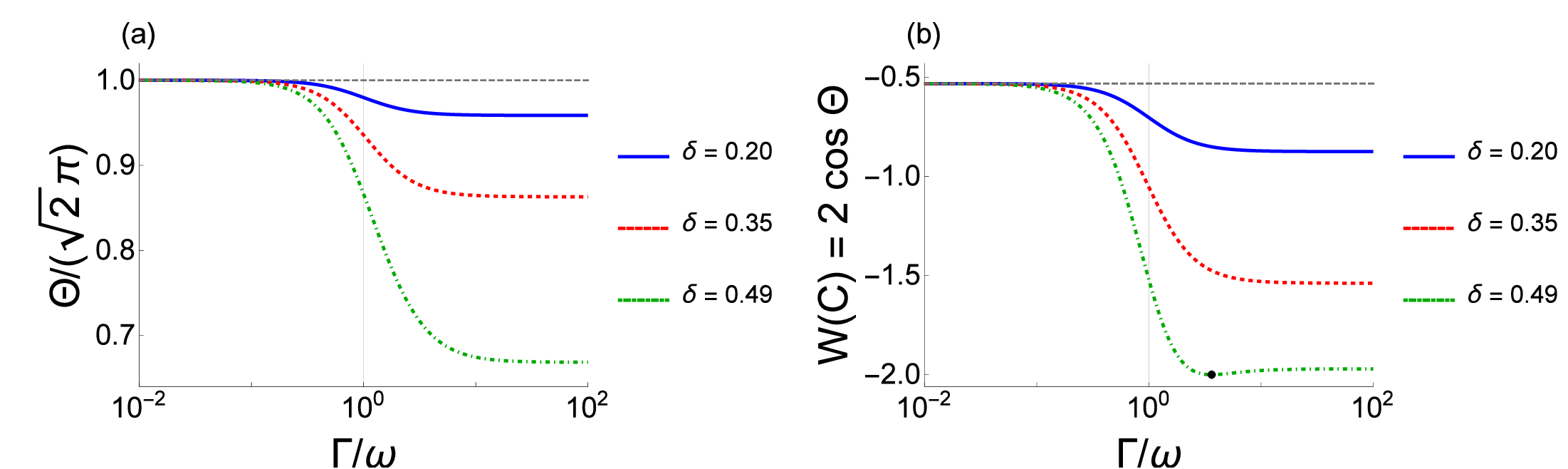}
\caption{(a) The holonomy angle $\Theta$ in units of its equal-weight value
$\sqrt2\pi$, and (b) the Wilson loop $W(C)=2\cos\Theta$, as functions of
$\Gamma/\omega$ for $\delta=0.2$, $0.35$, and $0.49$. Dashed lines mark
the Wilczek--Zee limits at $\Gamma/\omega\to0^+$. The dot in (b) marks the
minimum $W(C)=-2$ at $(\Gamma/\omega)_c\simeq3.6$ for $\delta=0.49$, where
the holonomy passes through $-\mathbb{I}_2$.}
\label{fig:Ec}
\end{figure*}

\subsubsection{Dissipative control and limiting cases}

The dissipative population current can be defined as
\begin{equation}
J_{\rm diss}(t)\equiv\Gamma[p(t)-p_{\rm eq}(t)]=-\dot p.
\end{equation}
so that for the prescribed support frame, a single measured trajectory $p(t)$
determines the instantaneous dressing factor and hence the holonomy. The causal chain
reservoir modulation $\to p(t)\to f(\phi)\to\mathcal{A}_{\text{U}}\to
\mathcal{U}_{\text{U}}(C)$ turns the spectral dressing of
Sec.~\ref{reductions} into a dynamical, tunable quantity rather than a
fixed property of the state. Three limits organize the behavior:

(i) \emph{Slow reservoir}, $\Gamma/\omega\to0^+$: in the periodic steady
state the weights collapse to $1/2$, $m\to0$, $E(0)=\pi/2$, and
\begin{equation}
\Theta\to\sqrt2\,\pi,\qquad
W(C)\to2\cos(\sqrt2\,\pi),
\end{equation}
the equal-weight Wilczek--Zee transport of Sec.~\ref{reductions}.

(ii) \emph{Fast reservoir}, $\Gamma/\omega\to\infty$: the system tracks
$p_{\rm eq}(t)$, $m\to4\delta^2$, and
\begin{equation}
\Theta\to2\sqrt2\,E(4\delta^2).
\end{equation}

(iii) \emph{Weak modulation}, $\delta\ll1$: expanding
$E(m)=\frac{\pi}{2}\bigl(1-\frac{m}{4}-\frac{3m^2}{64}-\cdots\bigr)$,
\begin{equation}
\Theta=\sqrt2\,\pi\left[
1-\frac{\delta^2\Gamma^2}{\Gamma^2+\omega^2}+O(\delta^4)
\right],
\end{equation}
so the deviation from the maximally mixed holonomy is of order $\delta^2$
and is maximized in the tracking regime $\Gamma\gg\omega$.
Figure~\ref{fig:Ec} shows $\Theta$ and $W(C)$ over the full range of
$\Gamma/\omega$. Read from left to right, the figure displays a crossover
from Wilczek--Zee transport at $\Gamma/\omega\to0^+$, where
$\mathcal{A}_{\text{U}}\to E^\dagger\dif E$, to dissipatively dressed
Uhlmann transport, in which the reservoir imprints the spectral modulation
on the holonomy.

Topologically the example is minimal: the orbit is a circle, and since
$H^2(S^1)=0$ there is no Chern class to carry. The nontrivial holonomy
is therefore a geometric monodromy of the closed support loop, rather
than a characteristic-class invariant. Its role in the hierarchy of
Table~\ref{tab:examples} is complementary to the Yang monopole: the support
loop is held fixed while the reservoir continuously tunes the spectral
dressing.

\subsection{Yang monopole as rank-2 Uhlmann holonomy}
\label{yang}

We now construct a physical realization with nontrivial second Chern
topology. Consider a four-level system with Hamiltonian
\begin{equation}
H(\mathbf{n})=\sum_{a=1}^5 n_a\Gamma_a,\quad
\mathbf{n}\in S^4,\; \sum_{a=1}^5 n_a^2=1,
\end{equation}
where we adopt the standard representation
$\Gamma_i=\sigma_1\otimes\sigma_i$ ($i=1,2,3$),
$\Gamma_4=\sigma_2\otimes I_2$, and $\Gamma_5=\sigma_3\otimes I_2$,
which satisfies $\{\Gamma_a,\Gamma_b\}=2\delta_{ab}$. Since $H^2=I$, the
spectrum is two-fold degenerate at $\pm1$. An explicit orthonormal frame for
the positive-energy subspace, regular on the patch $n_5<1$, is
\begin{align}
|u_1\rangle&=\frac{1}{\sqrt{2(1-n_5)}}
\begin{pmatrix}-n_3+\mi n_4\\ -n_1-\mi n_2\\ n_5-1\\ 0\end{pmatrix},
\notag\\
|u_2\rangle&=\frac{1}{\sqrt{2(1-n_5)}}
\begin{pmatrix}-n_1+\mi n_2\\ n_3+\mi n_4\\ 0\\ n_5-1\end{pmatrix},
\label{yang_frame}
\end{align}
satisfying $H|u_\alpha\rangle=|u_\alpha\rangle$ (as follows from
$\{\Gamma_a,\Gamma_b\}=2\delta_{ab}$) and
$\langle u_\alpha|u_\beta\rangle=\delta_{\alpha\beta}$; the corresponding
projector is $EE^\dagger=\frac12\left(I+\sum_a n_a\Gamma_a\right)$ with
$E=(|u_1\rangle,|u_2\rangle)$. We define the fixed-rank state
\begin{equation}
\rho(\mathbf{n})=\tfrac12 P_+(\mathbf{n}),\quad
P_+(\mathbf{n})=\tfrac12\left(I+\sum_{a=1}^5 n_a\Gamma_a\right),
\end{equation}
with eigenvalues $\{1/2,1/2,0,0\}$, so $\rho$ is a genuine rank-2 mixed
state, and the eigenvalue matrix in this frame is
$r=\diag(1/2,1/2)=\frac12 I_2$.

Because the nonzero eigenvalues are equal, the coefficient in
Eq.~\eqref{AU_explicit} equals one and the Uhlmann connection reduces to
the
Wilczek--Zee form $\mathcal{A}_{\text{U}}=E^\dagger\dif E$. Differentiating
Eq.~\eqref{yang_frame} and using $n_a\dif n_a=0$ gives the explicit
Yang-monopole potential
\begin{equation}
\mathcal{A}_{\text{U}}=\frac{\mi}{2}\sum_{a=1}^3 \mathcal{A}_{\text{U}}^a\,\sigma_a,
\label{yang_A}
\end{equation}
with
\begin{align}
\mathcal{A}_{\text{U}}^1&=\frac{n_4\dif n_1-n_3\dif n_2-n_1\dif n_4+n_2\dif n_3}{1-n_5},\notag\\
\mathcal{A}_{\text{U}}^2&=\frac{n_3\dif n_1+n_4\dif n_2-n_1\dif n_3-n_2\dif n_4}{1-n_5},\notag\\
\mathcal{A}_{\text{U}}^3&=\frac{n_1\dif n_2-n_2\dif n_1+n_4\dif n_3-n_3\dif n_4}{1-n_5}.
\label{yang_Aa}
\end{align}
Equivalently $\mathcal{A}_{\text{U}}^a=\bar\eta^a_{\mu\nu}n_\mu\dif n_\nu/(1-n_5)$, where
$\bar\eta^a_{ij}=\epsilon_{aij}$ ($i,j=1,2,3$),
$\bar\eta^a_{\mu4}=-\delta_{a\mu}$, $\bar\eta^a_{4\nu}=\delta_{a\nu}$ is
the 't Hooft symbol. In particular $\tr\mathcal{A}_{\text{U}}=0$: the
connection is traceless in this gauge, consistently with $c_1=0$. (The frame
\eqref{yang_frame} is singular at $n_5=1$, the usual Dirac string; a frame
regular near $n_5=1$ is obtained by $n_5\to-n_5$, and the two patches are
related by a gauge transformation.)

To compute the second Chern number we use the projector curvature identity
$\tr(\mathcal{F}_{\text{U}}\wedge\mathcal{F}_{\text{U}})
=\tr[P_+(\dif P_+)^4]$ (see Appendix~\ref{app:projector}) and the Clifford
algebra; since $\tr\mathcal{F}_{\text{U}}=0$ by Eq.~\eqref{yang_A}, the
general definition of Sec.~\ref{topology} reduces to this trace form. One
finds
\begin{equation}
\tr[P_+(\dif P_+)^4]
=\tfrac18\epsilon_{abcde}n_a\dif n_b\wedge\dif n_c\wedge\dif n_d\wedge\dif n_e
=3\,\omega_{S^4},
\end{equation}
with $\omega_{S^4}$ the volume form on the unit $4$-sphere. Since
$\mathlarger{\int}_{S^4}\omega_{S^4}=8\pi^2/3$,
\begin{equation}
C_2=\frac{1}{8\pi^2}\int_{S^4}
\tr(\mathcal{F}_{\text{U}}\wedge\mathcal{F}_{\text{U}})=1.
\end{equation}
The sign flips for $P_-$ or upon reversing the orientation of $S^4$.
Because $H^2(S^4,\mathbb{Z})=0$, the first Chern class vanishes and the
rank-2 bundle cannot split into a sum of line bundles. The topology is
\emph{purely non-Abelian}, detected solely by $C_2$. This is the natural
mixed-state analogue of the four-dimensional quantum Hall response, where
$C_2$ controls the quantized nonlinear transport coefficient; the Uhlmann
Wilson loop of the degenerate occupied subspace provides a direct route to
access this invariant.

We conclude by examining what the Uhlmann phase can see in this, the
topologically richest example of the paper. The analysis proceeds in three
steps.

\emph{Step 1: reduction to $SU(2)$.} The explicit gauge
\eqref{yang_A} is traceless, $\tr\mathcal{A}_{\text{U}}=0$, so
$\mathcal{U}_{\text{U}}(C)\in SU(2)$ for every loop. This is consistent with
the vanishing of $c_1$, although $c_1=0$ alone would not in general imply
such a reduction; here the reduction is explicit.

\emph{Step 2: reality of the amplitude.} The equal weights give
$r(0)=s^\dagger(\rho(0))s(\rho(0))=r=\frac12 I_2$, so the Uhlmann phase
reduces to
\begin{equation}
\theta_{\text{U}}=\arg\tr\mathcal{U}_{\text{U}}(C).
\end{equation}
Since $\mathcal{U}_{\text{U}}(C)\in SU(2)$, its eigenvalues form a conjugate
pair $\me^{\pm\mi\Theta(C)}$ with $\Theta(C)\in[0,\pi]$, the trace being
gauge invariant for a closed loop. Hence
\begin{align}
&\mathcal{Z}_{\text{U}}
=\tr[r(0)\mathcal{U}_{\text{U}}(C)]
=\cos\Theta(C)\in\mathbb{R},\notag\\
&\theta_{\text{U}}=\arg \mathcal{Z}_{\text{U}}=
\begin{cases}
0, & \Theta(C)<\pi/2,\\
\pi, & \Theta(C)>\pi/2.
\end{cases}
\end{align}
The angle $\Theta(C)$ is determined by the holonomy and varies continuously
with the loop, so $\theta_{\text{U}}$ jumps when $\Theta(C)$ crosses
$\pi/2$.

\emph{Step 3: what is and is not topological.} The restriction
$\theta_{\text{U}}\in\{0,\pi\}$ is kinematic rather than topological: it
records only the sign of the real number $\cos\Theta(C)$, a consequence of
the $SU(2)$ trace, and no integer invariant enters it. In particular, $c_1$
vanishes identically because $H^2(S^4,\mathbb{Z})=0$, so there is no Chern
parity for the phase to shadow. This contrasts sharply with the
$\mathbb{C}P^1$ example of Sec.~\ref{example_phase}, where, for the
equatorial loop, the value $\pi$ was pinned by $c_1(L_1)=-1$ and only the
choice between the two quantized values was left to the spectral weights. The
genuine invariant $C_2=\pm1$ is carried by the non-Abelian part of
$\mathcal{F}_{\text{U}}$ and never reaches the scalar trace; it is
characterized directly by the second Chern number of the underlying rank-two
bundle, with the Wilson loop providing the natural gauge-invariant probe,
although a detailed Wilson-loop characterization of $C_2$ lies beyond the
scope of the present work. Finally, this example is the minimal realization
of the splitting principle of Sec.~\ref{factorization}: a rank-1 family on
$S^4$ would carry no topology at all, since $c_1\in H^2(S^4)=0$ and higher
Chern classes require rank at least two, so the doubly degenerate equal-weight
state is the smallest rank supporting a non-factorizable second Chern number.

\begin{table*}[ht]
\centering
\caption{The four examples at a glance. ``Topology'' lists the characteristic
classes of the Uhlmann bundle over the parameter manifold; ``observable''
records the natural gauge-invariant probe, the scalar phase or the Wilson
loop, and its behavior.}
\label{tab:examples}
\footnotesize
\setlength{\tabcolsep}{3pt}
\begin{tabular}{l l l l l}
\hline
Example & Base & Spectrum & Topology & Uhlmann observable \\
\hline
$\mathbb{C}P^1$ (Sec.~\ref{example_phase}) & $S^2$ & $\lambda_1\neq\lambda_2$ & $c_1(L_1)=-1$, $c_1(L_2)=0$ & $\mathbb{Z}_2$ shadow (symmetric loop), weights select \\
Non-Abelian $N=3$ (Sec.~\ref{nonabelian_example}) & $(\phi,\chi,\theta)$ & $\lambda_1\neq\lambda_2$ & trivial, $c_1(L_a)=0$ & complex, continuous in weights and loop \\[6pt]
Dissipative orbit (Sec.~\ref{lindblad_example}) & $S^1$ orbit & $\lambda_{1,2}(t)$, reservoir-driven & none, $H^2(S^1)=0$ & \shortstack[l]{$\theta_{\text{U}}\equiv\pi$ (blind); $W(C)=2\cos[2\sqrt2 E(m)]$ \\set by $\Gamma/\omega$} \\
Yang monopole (Sec.~\ref{yang}) & $S^4$ & equal weights & $c_1=0$, $C_2=\pm1$ & $\{0,\pi\}$, kinematic, blind to $C_2$ \\
\hline
\end{tabular}
\end{table*}

\subsection{Summary and comparison of the examples}
\label{examples_summary}

The geometric content of the theory is not exhausted by the scalar Uhlmann
phase. The fundamental objects are the holonomy $\mathcal{U}_{\text{U}}(C)$
and the characteristic classes of the fixed-rank bundle, while
$\theta_{\text{U}}=\arg\tr[r(0)\mathcal{U}_{\text{U}}(C)]$ is a single
gauge-invariant projection of that data. The four examples of
Table~\ref{tab:examples}, together with the geometric-topological hierarchy
of Sec.~\ref{topology}, reveal a layered structure that can be summarized as
follows.

\subsubsection{Holonomy versus scalar phase}

As shown in Sec.~\ref{phase_vs_topology}, the full holonomy is homotopy
invariant if and only if the curvature vanishes, and on the simply connected
base $\mathcal{D}_k^N\simeq\operatorname{Gr}(k,N)$ a flat connection would
render all holonomies trivial. A nonvanishing $\theta_{\text{U}}$ therefore
necessarily requires nonvanishing curvature: the scalar phase is generically
geometric rather than topological, and its jumps mark geometric phase
transitions rather than changes of topological sector. Being a single number
extracted from the holonomy matrix, $\theta_{\text{U}}$ may in addition remain
insensitive to parts of that matrix's structure, a point to which we return
below.

\subsubsection{Where the topology resides}

The integer topology resides in the characteristic classes of the bundle and
in non-Abelian Wilson loops, not in $\theta_{\text{U}}$ itself. The bridge
between the two is
\begin{equation}
\tr\mathcal{F}_{\text{U}}=\tr F_K,
\quad\text{i.e.,}\quad
c_1(\text{Uhlmann bundle})=c_1(\mathcal{S}_k).
\end{equation}
How much of this topology the scalar phase captures is decided by the
structure of the connection, and three cases arise, as analyzed in
Sec.~\ref{phase_vs_topology}.

\emph{Case 1: $k=1$ (pure states).} The scalar phase reduces to the Berry
phase,
\begin{equation}
\theta_{\text{U}}\;\longrightarrow\;\gamma_{\text{B}}
=\mathlarger{\int}_\Sigma F,
\end{equation}
which for the rotationally symmetric representatives and equatorial loops
considered here equals
\begin{equation}
\gamma_{\text{B}}=\pi c_1\pmod{2\pi}.
\end{equation}
In such symmetry-adapted constructions $\theta_{\text{U}}$ is therefore
genuinely tied to the first Chern number. For a generic loop or curvature
distribution the Berry phase is a flux through a spanning surface and is no
longer fixed by $c_1$ alone.

\emph{Case 2: $k>1$ with an Abelian connection within $U(k)$.} As in the
rank-2 example of Sec.~\ref{example_phase}, the same conclusion survives in
$\mathbb{Z}_2$ form,
\begin{equation}
\theta_{\text{U}}\in\{0,\pi\},
\end{equation}
reflecting, within that construction, the parity of the twisted eigenline's
Berry holonomy: the scalar phase then carries a $\mathbb{Z}_2$ shadow of the
first Chern class.

\emph{Case 3: genuinely non-Abelian families ($k\ge2$).} Here
$\theta_{\text{U}}$ is a highly compressed functional of the full holonomy:
(i) it may vary with the non-Abelian curvature, as the Yang monopole
illustrates through its dependence on the $SU(2)$ angle $\Theta(C)$;
(ii) it retains too little information to reconstruct the holonomy matrix or
the higher characteristic classes, being insensitive to any part of the
holonomy that does not affect the weighted trace
$\tr[r(0)\mathcal{U}_{\text{U}}(C)]$; and (iii) the second Chern number is
invisible to it: the scalar phase does not determine $C_2$
(scalar $\theta_{\text{U}}\not\Rightarrow C_2$), whereas the Wilson loop
does. The obstruction is structural: by the splitting argument of
Sec.~\ref{factorization}, a non-factorizable $C_2$ requires spectral
degeneracy, which is precisely the regime in which the scalar phase loses the
ability to encode the full holonomy.

\subsubsection{The role of the spectrum}

A splitting argument makes the role of the spectrum explicit. For families
with non-degenerate spectrum, Proposition~\ref{prop:splitting} gives
\begin{equation}
c(E)=\prod_{a=1}^k\bigl(1+c_1(L_a)\bigr),
\end{equation}
so that every characteristic class of the fixed-rank bundle is a polynomial
in the Berry Chern data of the eigenlines, and no information beyond these
Abelian data remains. Genuinely non-factorizable non-Abelian topology
therefore requires spectral degeneracy, which fuses the eigenlines into a
single non-Abelian object: at equal weights the Uhlmann connection reduces,
as shown in Sec.~\ref{reductions}, to the Wilczek--Zee form
$E^\dagger\dif E$ on that fused subspace, precisely the doorway through
which higher Chern classes enter. The rank-2 Yang monopole on $S^4$ realizes
this minimally, with a quantized second Chern number $C_2=\pm1$ that no
splitting into line bundles can reproduce.

\subsubsection{Lessons from the four examples}

These structural observations translate into three physical lessons,
summarized in Table~\ref{tab:examples}. First, whenever the underlying
eigenlines carry band topology, the Uhlmann phase retains, for the
symmetry-adapted loops, a quantized value inherited from that topology and
loses it at a critical point of the control parameter, providing a
parameter-driven signature of the underlying topology; the $\mathbb{C}P^1$
example occupies this layer, with the spectral weights alone driving the
transition. Second, once the spectral weights become unequal, the phase
turns into a continuous, spectrum-dependent observable that probes the
coherence structure of the mixed state through the Uhlmann damping of the
interlevel matrix elements; the non-Abelian $N=3$ model occupies this layer,
its weight-dependent holonomy being genuinely distinct from the
Wilczek--Zee transport of the degenerate limit, while the dissipative qutrit
orbit makes this layer genuinely dynamical: a modulated reservoir steers the
spectral weights in time without ever leaving the rank-2 stratum, and the
Wilson loop $W(C)=2\cos[2\sqrt2 E(m)]$ records the reservoir's tracking
ratio $\Gamma/\omega$ at fixed modulation amplitude $\delta$ in closed
form. In doing so it exposes the scalar phase as maximally uninformative,
$\theta_{\text{U}}=\pi$ identically, and thereby delivers the clearest
illustration of the hierarchy
$\theta_{\text{U}}\subset\text{Wilson loop}\subset\text{characteristic
classes}$: a constant scalar phase coexisting with a continuously tuned,
genuinely nonequilibrium holonomy. Third,
where the topology is purely non-Abelian, as in the Yang monopole, the
scalar Uhlmann phase considered here does not encode $C_2$, and Wilson loops
become the relevant observable.

The division of labor among the four examples is sharp: the
$\mathbb{C}P^1$ model owns the quantized $0$-to-$\pi$ transition, the
non-Abelian $N=3$ model owns genuinely non-Abelian holonomy, the Yang
monopole owns integer topology, and the dissipative orbit owns
nonequilibrium, reservoir-controlled transport. In all cases the physical
picture is the same. The supporting subspace
determines the characteristic classes of the fixed-rank bundle; a
non-degenerate spectral resolution may carry additional eigenline topology
within it; the spectrum decides how much of this structure survives in a
given observable; and the Uhlmann phase records the transport of
purifications along the way. The hierarchy of observables is thus
$\theta_{\text{U}}\subset\text{Wilson loop}\subset\text{characteristic
classes}$, with the scalar phase capturing at most an Abelian or
$\mathbb{Z}_2$ shadow of the topology encoded in the full bundle.


\section{Conclusion}
\label{conclusion}

We have constructed Uhlmann's theory on the manifold of density matrices of
fixed rank below the Hilbert-space dimension, where the full-rank formalism
does not apply. Minimal purifications place a principal $U(k)$-bundle over
the rank-$k$ stratum and reduce the structure group from $U(N)$ to $U(k)$;
the Uhlmann connection is fixed uniquely by a Sylvester equation, and its
eigenframe expression interpolates between the Berry, Wilczek--Zee, and
faithful Uhlmann connections, while the Bures metric follows from the same
horizontal lifts.

The contrast with the full-rank theory is instructive. In the full-rank
theory the base manifold is contractible and the bundle is trivial:
holonomies are genuine geometric objects, but every loop is contractible, no
nontrivial characteristic class exists, and no integer invariant constrains
the holonomy, so topology offers no classification. In the fixed-rank theory
the base deformation-retracts onto the Grassmannian and the first Chern
class equals that of the tautological bundle; for families with
non-degenerate spectrum, a splitting argument shows that all characteristic
classes factorize into eigenline Berry data, so that non-factorizable higher
Chern topology requires failure of the global eigenline splitting, which in
the present spectral setting requires degeneracy. A rank-2 Yang monopole
realizes this minimally with a quantized second Chern number, linking the
mixed-state geometry to the four-dimensional quantum Hall response. The
scalar Uhlmann phase is generically geometric: in the symmetry-adapted
constructions considered here it carries at most a $\mathbb{Z}_2$ shadow of
the first Chern class, its jumps mark geometric phase transitions rather than
changes of topological sector, and higher invariants require non-Abelian
probes such as Wilson loops.

The solvable examples illustrate what this richer structure may offer
physically: a spectral-weight-driven $0$-to-$\pi$ transition of an Abelian
rank-2 phase; a genuinely non-Abelian, weight-dependent holonomy on a
topologically trivial bundle; a dissipatively driven qutrit orbit in which a
modulated reservoir steers the spectral weights and dresses the Uhlmann
holonomy through a closed elliptic-integral factor; and a Yang monopole
whose $C_2$ lies beyond the reach of any scalar phase. Because purifications
require only a $k$-dimensional ancilla,
these phases and holonomies are accessible to interferometric measurement on
a smaller platform than the full-rank protocol demands. Developing a
Wilson-loop characterization of the second Chern number, and extending the
dissipative construction to reservoirs that act on the kernel and can close
the spectral gap, are natural next steps.

\begin{acknowledgments}
H.G. was supported by the Quantum Science and Technology-National Science and
Technology Major Project (Grant No.~2021ZD0301904) and the National Natural
Science Foundation of China (Grant No.~12447216). X.-Y. H was supported by the National Natural Science Foundation of China (Grant No. 12405008).
\end{acknowledgments}

\appendix
\section{Details of the Ehresmann connection}
\label{app:connection}

In this appendix we collect the technical verifications omitted in
Sec.~\ref{connection}.

\subsection{Invertibility of the Sylvester operator}

Let $h=W^\dagger W>0$ and consider the linear map
$L:\omega\mapsto h\omega+\omega h$ on $k\times k$ matrices. In an eigenbasis
of $h$ with eigenvalues $\lambda_a>0$, one has
$L(\omega)_{ab}=(\lambda_a+\lambda_b)\omega_{ab}$. Since
$\lambda_a+\lambda_b>0$ for all $a,b$, $L$ is invertible. Moreover, $L$
commutes with Hermitian conjugation, $L(\omega^\dagger)=L(\omega)^\dagger$,
and therefore preserves the subspace of anti-Hermitian matrices. The
right-hand side of Eq.~\eqref{eq:sylvester} is anti-Hermitian, so its unique
preimage $\omega(X)$ is anti-Hermitian, as required. This property would fail
for a singular amplitude, where zero eigenvalues of $h$ would make $L$
non-invertible; this is why minimal purifications are essential.

\subsection{Equivariance under right translation}

Let $W'=WV$ with constant $V\in U(k)$. A tangent vector transforms as
$X\mapsto XV$. From Eq.~\eqref{eq:decomp} applied at $W$,
\begin{equation}
XV = X_H V + WV\,(V^\dagger\omega(X)V).
\end{equation}
The term $X_HV$ is horizontal because
\begin{align}
&(X_HV)^\dagger(WV)=V^\dagger X_H^\dagger W V\notag\\
=&V^\dagger W^\dagger X_H V = (WV)^\dagger(X_HV).
\end{align}
Since $V^\dagger\omega(X)V\in\mathfrak{u}(k)$ and the decomposition at $WV$
is
unique, we obtain
\begin{equation}
\omega_{WV}(XV)=V^\dagger\omega_W(X)V.
\end{equation}
This proves $R_V^*\omega=\mathrm{Ad}_{V^{-1}}\omega$.

\subsection{Gauge transformation law}

For a local gauge transformation $W'=WV$ with varying $V$, one has
$h'=W'^\dagger W'=V^\dagger h V$ and
$\dif W'=(\dif W)V+W\dif V$. The left-hand side of Eq.~\eqref{eq:sylvester}
transforms as
\begin{align}
&W'^\dagger\dif W'-\dif W'^\dagger W'
=V^\dagger\bigl(W^\dagger\dif W-\dif W^\dagger W\bigr)V\notag\\
  &\qquad\qquad\qquad\qquad\quad+V^\dagger h\dif V-\dif V^\dagger h V\notag\\
=&h'\bigl(V^\dagger\omega V\bigr)+\bigl(V^\dagger\omega V\bigr)h'+h'\bigl(V^\dagger\dif V\bigr)-\bigl(V^\dagger\dif V\bigr)^\dagger h'
  \notag\\
=&h'\omega'+\omega'h',
\end{align}
where in the last step we have set
\begin{equation}
\omega'=V^\dagger\omega V+V^\dagger\dif V
\end{equation}
and used $(V^\dagger\dif V)^\dagger=-V^\dagger\dif V$. Uniqueness of the
solution of Eq.~\eqref{eq:sylvester} implies that $\omega'$ is the connection
form at $W'$.

\subsection{Intertwining structure}

The reproducing property $\omega(u^\#)=u$ and the equivariance
$R_V^*\omega=\mathrm{Ad}_{V^{-1}}\omega$ combine into a single intertwining
relation. From the pushforward identity
\begin{equation}
R_{V*}u^\#_W=(V^\dagger uV)^\#_{WV},
\label{pushforward}
\end{equation}
which follows from $W\me^{tu}V=WV\me^{tV^\dagger uV}$, one obtains
\begin{equation}
\omega_{WV}\bigl((V^\dagger uV)^\#_{WV}\bigr)
= V^\dagger uV
= V^\dagger\omega_W(u^\#_W)V.
\label{intertwining}
\end{equation}
Thus $\omega$ intertwines the natural action on $T\mathcal{P}_k$,
$X\mapsto XV$, with the adjoint action on $\mathfrak{u}(k)$,
$u\mapsto V^\dagger uV$; equivalently, the following diagram commutes:
\begin{center}
\includegraphics[width=3.41in,clip]{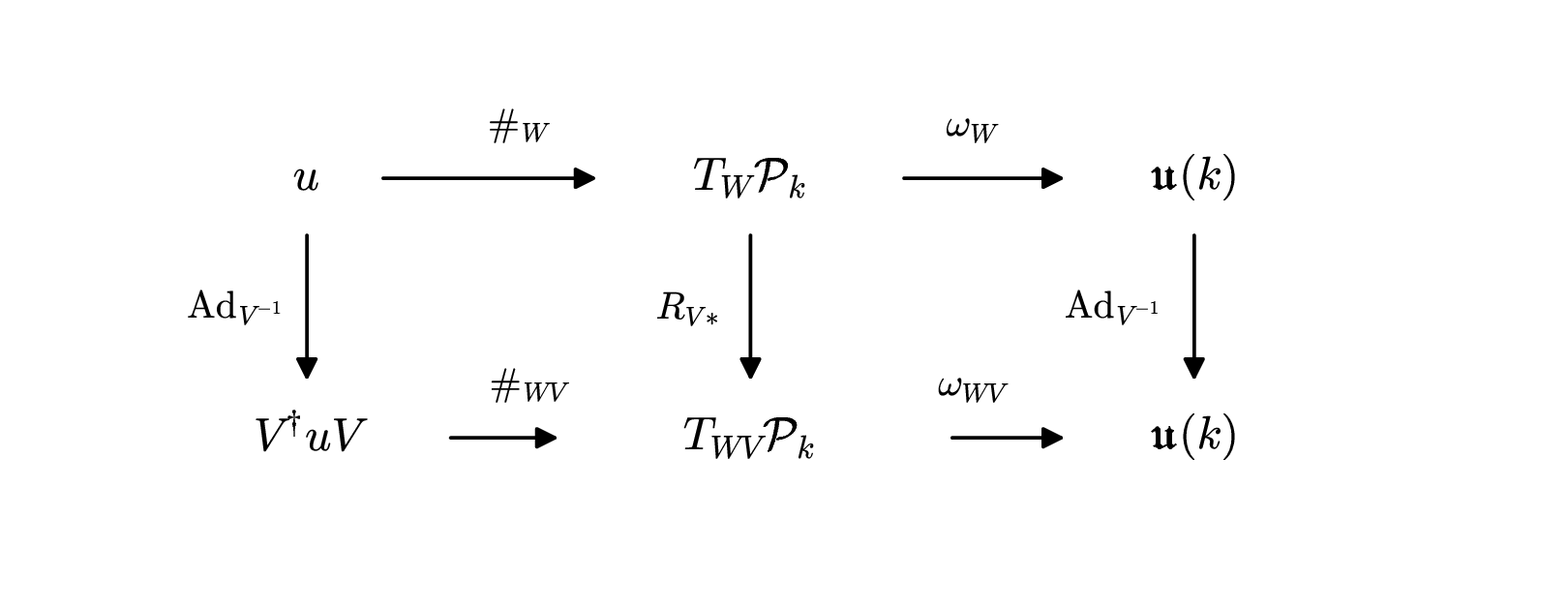}
\captionof{figure}{The Ehresmann connection $\omega$ intertwines the natural
$U(k)$ action on $T\mathcal{P}_k$ with the adjoint action on
$\mathfrak{u}(k)$, cf.\ Eq.~\eqref{intertwining}.}
\label{Fig1}
\end{center}
Consequently, $\ker\omega_{WV}=R_{V*}\ker\omega_W$, so the horizontal
distribution $H_W\mathcal{P}_k=\ker\omega_W$ is smooth and $U(k)$-equivariant.
On a single fiber, where $\omega$ reduces to the Maurer--Cartan form, this is
the usual alignment of right translations with the adjoint representation.

\subsection{Horizontal transport equivalence}

Let $\tilde\gamma(t)=W(t)$ be a curve with tangent $\tilde X=\dot W$. From
Eq.~\eqref{eq:decomp},
\begin{equation}
\dot W=\dot W_H+W\omega(\dot W),
\end{equation}
where $\dot W_H$ is horizontal. A direct calculation gives
\begin{equation}
W^\dagger\dot W-\dot W^\dagger W
= h\,\omega(\dot W)+\omega(\dot W)h.
\end{equation}
Since $L$ is bijective, the left-hand side vanishes if and only if
$\omega(\dot W)=0$. Therefore
\begin{equation}
W^\dagger\dot W=\dot W^\dagger W
\Longleftrightarrow
\omega(\dot W)=0,
\end{equation}
which is Eq.~\eqref{eq:horizontal_equivalence}.

\section{Derivation of the Bures metric}
\label{app:bures}

We first note that the infimum in the Bures distance may be restricted to
minimal purifications. Indeed, any purification $W'$ of $\rho$ has the form
$W'=WC$ with $W$ minimal and $CC^\dagger=I_k$ (a co-isometry), and for two
purifications $W_i'=W_iC_i$ a short calculation gives
\begin{align}
&\tr\bigl[(W_1C_1-W_2C_2)(W_1C_1-W_2C_2)^\dagger\bigr]\notag\\
=&2-2\operatorname{Re}\tr\bigl[W_2^\dagger W_1\,C_1C_2^\dagger\bigr]
\ge 2-2F,
\end{align}
because $C_1C_2^\dagger$ is a contraction and
\begin{align}
&\max_{U\in U(k)}\operatorname{Re}\tr[W_2^\dagger W_1U]
=\tr\sqrt{W_1^\dagger W_2W_2^\dagger W_1}\notag\\
=&\tr\sqrt{\sqrt{\rho_1}\rho_2\sqrt{\rho_1}}=F,
\end{align}
and equality is attained within the minimal class by parallel purifications.

Let $\rho(t)$ be a curve in $\mathcal{D}_k^N$ and $W(t)$ a horizontal lift.
The Hilbert--Schmidt speed is
$\|\dot W\|_{\mathrm{HS}}^2=\tr(\dot W^\dagger\dot W)$. We evaluate this in
the eigenbasis of $\rho$, allowing for vanishing eigenvalues. Let
$|i\rangle$ denote the eigenstates of $\rho$ with eigenvalues $\lambda_i$,
where $i=1,\dots,N$ and $\lambda_i=0$ for $i>k$. A minimal purification is
an
$N\times k$ matrix whose columns are labeled by $j=1,\dots,k$. From
$\dot\rho=\dot WW^\dagger+W\dot W^\dagger$ and the explicit form of $W$ in
the eigenframe, one finds
\begin{align}
&\langle i|\dot\rho|j\rangle
=
\sqrt{\lambda_j}\,\langle i|\dot W|j\rangle
+\sqrt{\lambda_i}\,\bigl(\langle j|\dot W|i\rangle\bigr)^*
\quad(i\le k),\notag\\
&\langle i|\dot\rho|j\rangle
=
\sqrt{\lambda_j}\,\langle i|\dot W|j\rangle
\quad(i>k),
\end{align}
for $j\le k$. The horizontal condition $W^\dagger\dot W=\dot W^\dagger W$,
restricted to the support block, reads
\begin{equation}
\sqrt{\lambda_i}\,\langle i|\dot W|j\rangle
=
\sqrt{\lambda_j}\,\bigl(\langle j|\dot W|i\rangle\bigr)^*
\qquad(i,j\le k),
\end{equation}
so that
\begin{equation}
\langle i|\dot\rho|j\rangle
=
\frac{\lambda_i+\lambda_j}{\sqrt{\lambda_j}}\,\langle i|\dot W|j\rangle
\qquad(i,j\le k).
\end{equation}
Solving for $\langle i|\dot W|j\rangle$ and summing the modulus squares over
all $i=1,\dots,N$ and $j=1,\dots,k$ gives
\begin{equation}
\|\dot W\|_{\mathrm{HS}}^2
=
\sum_{i=1}^k\sum_{j=1}^k
\frac{\lambda_j\,|\langle i|\dot\rho|j\rangle|^2}{(\lambda_i+\lambda_j)^2}
+
\sum_{i=k+1}^N\sum_{j=1}^k
\frac{|\langle i|\dot\rho|j\rangle|^2}{\lambda_j}.
\end{equation}
In the first term, the summand is symmetric under $i\leftrightarrow j$, so it
equals one half of the same sum with $\lambda_j$ replaced by
$\lambda_i+\lambda_j$; the second term is
$|\langle i|\dot\rho|j\rangle|^2/(\lambda_i+\lambda_j)$ since $\lambda_i=0$
there. Hence
\begin{equation}
\|\dot W\|_{\mathrm{HS}}^2
=
\frac12\sum_{\lambda_i+\lambda_j>0}
\frac{|\langle i|\dot\rho|j\rangle|^2}{\lambda_i+\lambda_j},
\end{equation}
because $\lambda_i+\lambda_j>0$ whenever $j\le k$, even if $\lambda_i=0$.
This is precisely the Bures line element \eqref{bures_metric}, confirming
the
consistency.

\section{Local sections, pullback, and the Sylvester equation}
\label{app:local}

\subsection{Local sections and their non-global nature}

The section $s(\rho)=E\sqrt{r}$ is defined only locally because the
eigenframe
$E$ is determined by $\rho$ only up to the allowed rotations and no smooth
global choice exists when the base $\mathcal{D}_k^N$ deformation-retracts
onto the Grassmannian $\operatorname{Gr}(k,N)$, whose frame bundle is
nontrivial. Hence $s$ is a local section patched by transition functions of
the frame bundle, while the holonomy and the Uhlmann phase remain globally
defined and gauge covariant. In the full-rank case $k=N$,
$\sigma(\rho)=\sqrt{\rho}$ is global and the bundle is trivial.

\subsection{Pullback derivation of the Uhlmann connection}

Within the domain of $s$, write $W=s(\rho)\mathcal{U}$ with
$\mathcal{U}\in U(k)$. The relation
between the Ehresmann connection and the pulled-back connection,
\begin{equation}
\omega=\mathcal{U}^\dagger \pi^*\mathcal{A}_{\text{U}}\mathcal{U}
+\mathcal{U}^\dagger \dif_P \mathcal{U},
\label{app:omegaV}
\end{equation}
follows from the reproducing property and equivariance of $\omega$. To see
this, consider a curve $W(t)=s(\rho(t))\mathcal{U}(t)$ with tangent
$\dot W=\dif s(\dot\rho)\mathcal{U}+s\dot{\mathcal{U}}$.

For the fiber direction, take $\mathcal{U}(t)=\mathcal{U} \me^{tu}$, so that
$\dot{\mathcal{U}}=\mathcal{U}u$ and
$s\dot{\mathcal{U}}=s\mathcal{U}u=Wu=u^\#_W$. The reproducing property
$\omega(u^\#)=u$ then gives
$\omega(s\dot{\mathcal{U}})=u=\mathcal{U}^\dagger \dot{\mathcal{U}}$. By
linearity this extends to arbitrary
$\dot{\mathcal{U}}\in\mathfrak{u}(k)$:
\begin{equation}
\omega(s\dot{\mathcal{U}})=\mathcal{U}^\dagger \dot{\mathcal{U}}.
\end{equation}
For the base direction, at the point $s(\rho)$ we have
$\omega_s(\dif s(\dot\rho))=\mathcal{A}_{\text{U}}(\dot\rho)$. Under the
right
translation $R_V:s\mapsto sV$, the pushforward of $\dif s(\dot\rho)$ is
$\dif s(\dot\rho)V$, and equivariance
$R_V^*\omega=\mathrm{Ad}_{V^{-1}}\omega$ gives
\begin{equation}
\omega_{sV}\bigl(\dif s(\dot\rho)V\bigr)
=V^\dagger \omega_s(\dif s(\dot\rho))V
=V^\dagger \mathcal{A}_{\text{U}}(\dot\rho)V.
\end{equation}
Adding the two contributions and using
$\dot\rho=\pi_*\dot W$ and $\dot{\mathcal{U}}=\dif_P \mathcal{U}(\dot W)$, we
obtain
\begin{equation}
\omega(\dot W)
=\mathcal{U}^\dagger(\pi^*\mathcal{A}_{\text{U}})(\dot W)\mathcal{U}
+\mathcal{U}^\dagger(\dif_P \mathcal{U})(\dot W).
\end{equation}
Since this holds for arbitrary $\dot W$, Eq.~\eqref{app:omegaV} follows.
Substituting Eq.~\eqref{W=sV} into the Sylvester equation
\eqref{eq:sylvester} with $X=\dif W$ and using
$h=W^\dagger W=\mathcal{U}^\dagger r\mathcal{U}$, one finds
\begin{align}
&\mathcal{U}^\dagger\bigl(s^\dagger\dif s-\dif s^\dagger s\bigr)\mathcal{U}
+\mathcal{U}^\dagger r\dif \mathcal{U}-\dif \mathcal{U}^\dagger r\mathcal{U}\notag\\
=&\mathcal{U}^\dagger\bigl(r\,\pi^*\mathcal{A}_{\text{U}}
+\pi^*\mathcal{A}_{\text{U}}r\bigr)\mathcal{U}
+\mathcal{U}^\dagger r\dif \mathcal{U}
+\mathcal{U}^\dagger\dif \mathcal{U}\,\mathcal{U}^\dagger r\mathcal{U}.
\label{app:substituted}
\end{align}
Here we have used Eq.~\eqref{app:omegaV} to replace $\omega$ on the left-hand
side by $\pi^*\mathcal{A}_{\text{U}}$ and $\dif \mathcal{U}$. The terms
$\mathcal{U}^\dagger r\dif \mathcal{U}$ cancel between the two sides, and
the remaining
$\dif \mathcal{U}$-terms combine as
\begin{align}
-\dif \mathcal{U}^\dagger r\mathcal{U}
-\mathcal{U}^\dagger\dif \mathcal{U}\,\mathcal{U}^\dagger r\mathcal{U}
&=-\bigl(\dif \mathcal{U}^\dagger \mathcal{U}
+\mathcal{U}^\dagger\dif \mathcal{U}\bigr)\mathcal{U}^\dagger r\mathcal{U}\notag\\
&=-\dif(\mathcal{U}^\dagger \mathcal{U})\,\mathcal{U}^\dagger r\mathcal{U}=0,
\end{align}
where we have used $\dif(\mathcal{U}^\dagger \mathcal{U})=0$. Contracting
with an arbitrary
horizontal vector $\tilde{X}$ and using
$\pi^*\mathcal{A}_{\text{U}}(\tilde{X})=\mathcal{A}_{\text{U}}(X)$, we obtain
Eq.~\eqref{AU_sylvester}. To evaluate the right-hand side, note that the
frame
connection is $K=E^\dagger \dif E$ with $K^\dagger=-K$. From
$\dif s=(\dif E)\sqrt{r}+E\,\dif\sqrt{r}$ one computes
\begin{equation}
s^\dagger \dif s-\dif s^\dagger s
=2\sqrt{r}K\sqrt{r}+[\sqrt{r},\dif\sqrt{r}].
\end{equation}
Since $r$ is diagonal, $[\sqrt{r},\dif\sqrt{r}]=0$, and
Eq.~\eqref{AU_sylvester} reduces to
\begin{equation}
r\mathcal{A}_{\text{U}}+\mathcal{A}_{\text{U}}r
=2\sqrt{r}K\sqrt{r},
\end{equation}
whose components give Eq.~\eqref{AU_explicit}. The cancellation of the
$\dif \mathcal{U}$-terms relies solely on $\dif(\mathcal{U}^\dagger
\mathcal{U})=0$, i.e., on the unitarity
of the fiber coordinate; this is the fixed-rank counterpart of the
cancellation
of the $\dif U$-terms in the full-rank case.

\section{Chern-number evaluations for the topological examples}
\label{app:chern}

This appendix collects the two Chern-number calculations quoted in the main
text: the second Chern number of the non-degenerate $S^2\times S^2$ family
of
Sec.~\ref{factorization}, and the projector-curvature evaluation behind the
Yang-monopole invariant of Sec.~\ref{yang}.

\subsection{Second Chern number of the $S^2\times S^2$ family}
\label{app:s2xs2}

For the family defined in Sec.~\ref{factorization}, the curvature is block
diagonal, $\mathcal{F}_{\text{U}}=\diag(F_p,F_q)$, where
\begin{equation}
F_p=\frac{\mi}{2}\sin\theta_p\,\dif\theta_p\wedge\dif\phi_p,
\quad
F_q=\frac{\mi}{2}\sin\theta_q\,\dif\theta_q\wedge\dif\phi_q
\end{equation}
are the Berry curvatures of the two eigenlines $L_1,L_2$ pulled back from
the
factor spheres. Since each of $F_p^2$ and $F_q^2$ is a four-form pulled back
from a two-dimensional sphere,
\begin{equation}
\tr(\mathcal{F}_{\text{U}}\wedge\mathcal{F}_{\text{U}})
=F_p^2+F_q^2=0,
\end{equation}
while, using $F_p^2=F_q^2=0$ and Fubini's theorem,
\begin{align}
&\tr \mathcal{F}_{\text{U}}\wedge\tr \mathcal{F}_{\text{U}}
=(F_p+F_q)\wedge(F_p+F_q)=2F_p\wedge F_q,\notag\\
&\int_{S^2\times S^2}F_p\wedge F_q
=\left(\int_{S^2_p}F_p\right)\left(\int_{S^2_q}F_q\right)
=-4\pi^2,
\end{align}
because $\int_{S^2}F_{p,q}=2\pi\mi$, i.e.,
$c_1(L_{1,2})=\frac{\mi}{2\pi}\int F_{p,q}=-1$. Therefore
\begin{equation}
C_2=\frac{1}{8\pi^2}\bigl[0-2\cdot(-4\pi^2)\bigr]=1.
\end{equation}
We also note that $\tr \mathcal{F}_{\text{U}}=F_p+F_q$ integrates to
$2\pi\mi$
on each factor sphere, so that its pairing with either factor gives
$\frac{\mi}{2\pi}\mathlarger{\int}(F_p+F_q)=-1$, consistently with
$c_1(L_1)=c_1(L_2)=-1$
and with $\tr \mathcal{F}_{\text{U}}=\tr \mathcal{F}_K$.

\subsection{Projector curvature identity and the Yang monopole}
\label{app:projector}

For a projector $P$ satisfying $P^2=P$, the projected connection
$\nabla=P\dif$ has curvature $F=P(\dif P)^2P$, which in a local frame $E$
with $EE^\dagger=P$ and $\mathcal{A}=E^\dagger\dif E$ satisfies
\begin{equation}
\tr(F\wedge F)=\tr[P(\dif P)^4].
\end{equation}
This follows from the identity $P\dif P\,P=0$ and the cyclic property of the
trace. For the Yang-monopole projector
$P_+=\frac{1}{2}(I+\sum_a n_a\Gamma_a)$ we have
$\dif P_+=\frac{1}{2}\sum_a\dif n_a\,\Gamma_a$, and therefore
\begin{equation}
(\dif P_+)^4=\frac{1}{16}\sum_{a,b,c,d}\dif n_a\wedge\dif n_b\wedge\dif n_c\wedge\dif n_d\,\Gamma_a\Gamma_b\Gamma_c\Gamma_d.
\end{equation}
The relevant Clifford traces are
\begin{align}
&\tr(\Gamma_a\Gamma_b\Gamma_c\Gamma_d)
=4\bigl(\delta_{ab}\delta_{cd}-\delta_{ac}\delta_{bd}+\delta_{ad}\delta_{bc}\bigr),
\notag\\
&\tr(\Gamma_e\Gamma_a\Gamma_b\Gamma_c\Gamma_d)
=\pm4\,\epsilon_{eabcd},
\end{align}
with the sign in the second line fixed by the choice of Clifford
representation. Substituting into $\tr[P_+(\dif P_+)^4]$ and using
$\tr P_+=2$,
\begin{align}
&\tr[P_+(\dif P_+)^4]\notag\\
=&\frac{1}{32}\sum_{a,b,c,d}\dif n_a\wedge\dif n_b\wedge\dif n_c\wedge\dif n_d\,
\big[\tr(\Gamma_a\Gamma_b\Gamma_c\Gamma_d)\notag\\&+n_e\tr(\Gamma_e\Gamma_a\Gamma_b\Gamma_c\Gamma_d)\big].
\end{align}
The four-Gamma trace is symmetric under pairwise exchange of indices and
hence cancels against the fully antisymmetrized wedge product, while the
five-Gamma trace gives
\begin{align}
\tr[P_+(\dif P_+)^4]
&=\pm\frac{4}{32}\,\epsilon_{eabcd}\,n_e\,
\dif n_a\wedge\dif n_b\wedge\dif n_c\wedge\dif n_d
\notag\\&=\pm 3\,\omega_{S^4},
\end{align}
since
$\omega_{S^4}=\frac{1}{4!}\epsilon_{eabcd}n_e\dif n_a\wedge\dif n_b\wedge\dif n_c\wedge\dif n_d$
is the volume form of the unit $4$-sphere. With the orientation convention
of
Sec.~\ref{yang} the sign is positive, giving
$\tr[P_+(\dif P_+)^4]=3\,\omega_{S^4}$.

\bibliography{Review,Review1,Review2,Review3,Review4,Review5,Review6}
\end{document}